\documentclass[11pt]{article}
\usepackage{graphicx}
\usepackage{latexsym}
\usepackage{amsfonts}
\usepackage{color}
\usepackage[compact]{titlesec}
\usepackage{amsmath, amsthm}
\usepackage{indentfirst}
\usepackage{mathrsfs}
\usepackage{amscd}
\usepackage{amsfonts}
\usepackage{epsfig}
\usepackage{color}
\usepackage{graphics}
\usepackage{dsfont}
\usepackage{bbm}
\usepackage{fullpage}
\usepackage[numbers]{natbib}
\usepackage{authblk}
\usepackage{setspace}
\usepackage{algpseudocode}
\usepackage{algorithm}
\floatstyle{plaintop}
\restylefloat{algorithm}
\usepackage{caption}
\usepackage{rotating}
\usepackage{multirow}
\usepackage{tabularx}
\usepackage{lscape}
\usepackage{array}
\usepackage{float}
\usepackage{url}
\usepackage{tikz}
\allowdisplaybreaks[1]

\numberwithin{equation}{section}

\newtheorem{theorem}{Theorem}[section]
\newtheorem{lemma}[theorem]{Lemma}
\newtheorem{proposition}[theorem]{Proposition}
\newtheorem{corollary}[theorem]{Corollary}

\newtheorem{remark}[theorem]{Remark}

\newtheorem{definition}[theorem]{Definition}
\newcommand{\Var}{{\rm Var}}

\newcommand{\Go}{\Rightarrow}

\newcommand{\clz}{\hat{Z}}
\newcommand{\IG}{{\rm IG}}
\newcommand{\NIG}{{\rm NIG}}

\title{\bf A Stochastic Model of Order Book Dynamics using Bouncing Geometric Brownian Motions}
\author[1]{Xin Liu\thanks{Email: xliu9@clemson.edu.}}
\author[2]{Vidyadhar G. Kulkarni\thanks{Email: vkulkarn@email.unc.edu.}}
\author[2]{Qi Gong\thanks{Email: qgong@email.unc.edu.}}
\affil[1]{Department of Mathematical Sciences, Clemson University, Clemson, SC 29634.}
\affil[2]{Department of Statistics and Operations Research, University of North Carolina,
            Chapel Hill, NC 27599.}

\begin{document}
\maketitle

\begin{abstract}

We consider a limit order book, where buyers and sellers register to trade a security at specific prices. The largest price buyers on the book are willing to offer is called the market bid price, and the smallest price sellers on the book are willing to accept is called the market ask price. Market ask price is always greater than market bid price, and these prices move upwards and downwards due to new arrivals, market trades, and cancellations. We model these two price processes as ``bouncing geometric Brownian motions (GBMs)'', which are defined as exponentials of two mutually reflected Brownian motions. We then modify these bouncing GBMs to construct a discrete time stochastic process of trading times and trading prices, which is parameterized by a positive parameter $\delta$. Under this model, it is shown that the inter-trading times are inverse Gaussian distributed, and the logarithmic returns between consecutive trading times follow a normal inverse Gaussian distribution. Our main results show that the logarithmic trading price process is a renewal reward process, and under a suitable scaling, this process converges to a standard Brownian motion as $\delta\to 0$. We also prove that the modified ask and bid processes approach the original bouncing GBMs as $\delta\to0$. Finally, we derive a simple and effective prediction formula for trading prices, and illustrate the effectiveness of the prediction formula with an example using real stock price data.

\end{abstract}

\section{Introduction}

In a modern order-driven trading system, limit-sell and limit-buy orders arrive with specific prices and they are registered in a {\em limit order book (LOB)}. The price at which a buyer is willing to buy is called the bid price and the price at which a seller is willing to sell is called the ask price. The order book organizes the orders by their prices and by their arrival times within each price. The highest bid price on the book is called the market bid price, and the lowest ask price on the book is called the market ask price. In contrast to the limit orders, market orders have no prices: a market buy order is matched and settled against the sell order at the market ask price and a market sell order is matched and settled against the buy order at the market bid price. (We are ignoring the sizes of the orders in this simplified discussion.) When the market bid price equals the market ask price, a trade occurs, and the two matched traders are removed from the LOB. Thus immediately after the trade the market bid price decreases and the market ask price increases. Clearly the market ask price is always above the market bid price. Between two trading times, the market ask and bid prices fluctuate due to new arrivals, cancellations, market trades, etc.

There is an extensive literature on models of LOBs, including statistical analysis and stochastic modeling. In particular, Markov models have been developed in \cite{aj13, bhs08, cst10, cl10, cl13, hp15a, hp15b, k03}, to name a few. In such models, point processes are used to model arrival processes of limit and market orders, and the market bid and ask prices are formulated as complex jump processes. To simplify such complexity, one tries to develop suitable approximate models. Brownian motion type approximations are established, for example, in \cite{aj13, bhs08, cl10}, and law of large numbers is recently studied in \cite{hp15a, hp15b}.

It is clear that the stochastic evolution of the market ask and bid prices is a result of a complex dynamics of the trader behavior and the market mechanism. It makes sense to ignore the detailed dynamics altogether and directly model the market ask and bid prices as stochastic processes. We let $A(t)$ and $B(t)$ be the market ask and bid prices at time $t$, respectively, and model $\{A(t), t \ge 0\}$ and $\{B(t), t \ge 0\}$ as two stochastic processes with continuous sample paths that bounce off of each other as follows. Initially $A(0) > B(0)$. Intuitively, we assume that the market bid and ask prices evolve according to two independent geometric Brownian motions (GBMs), and bounce off away from each other whenever they meet. Hence we call this the ``bouncing GBMs" model of the LOB. To the best of our knowledge, this is the first time such a model is used to describe the dynamics of the market ask and bid price processes in the LOB.

%

Bouncing GBMs can be constructed from bouncing BMs (the detailed construction is given Section \ref{sec:MF}). Bouncing BMs have been studied by Burdzy and Nualart in \cite{bn01}, and a related model of bouncing Brownian balls has been studied by Saisho and Tanaka in \cite{st85}. Of these two papers, the one by Burdzy and Nuaalart is most relevant to our model. They study two Brownian motions in which  the lower one is reflected downward from the upper one. Thus the upper process is unperturbed by the lower process, while the lower process is pushed downward (by an appropriate reflection map) when it hits the upper process. We use a similar construction in our bouncing GBM model, except that in our case both processes reflect off of each other in opposite directions whenever they meet. We assume that the reflection is symmetric, which will be made precise in Section \ref{sec:MF}.

We would like to say that a transaction occurs when the market ask price process meets the market bid price process, and the transaction price is the level where they meet. Unfortunately, the bouncing GBMs will meet at uncountably many times in any finite interval of time. This will create uncountably many transactions over a finite interval of time, which is not a good model of the reality. In reality, transactions occur at discrete times. Denote by $T_n$ the $n$-th transaction time, and $P_n$ the price at which the $n$-th transaction is settled. We are interested in studying the discrete time process $\{(P_n,T_n), n \ge 1\}$. To define this correctly and conveniently, we assume a price separation parameter $\delta > 0$, and construct two modified market ask and bid price processes $A_\delta$ and $B_\delta$ from the bouncing GBMs $A$ and $B$. One can think of $\delta$ as representing the tick size of the LOB, typically one cent. The construction of $A_\delta$ and $B_\delta$ enables us to define a discrete time stochastic process $\{(P_{\delta,n}, T_{\delta,n}), n \ge 1\}$ of transaction prices and times. The precise definitions of $A_\delta, B_\delta$ and $(P_{\delta,n}, T_{\delta,n})$ are given in Section \ref{sec:ttp}.

We show that the inter-trading times $T_{\delta,n+1}-T_{\delta,n}$ follow an inverse Gaussian (IG) distribution, and the logarithmic return between consecutive trading times $\ln(P_{\delta,n+1}/P_{\delta,n})$ follow a normal inverse Gaussian (NIG) distribution. We then formulate the logarithmic trading price process as a renewal reward process in terms of inter-trading times and successive logarithmic returns. It is worth noting that $\delta$ is typically small, and in the numerical example in Section \ref{numerical}, $\delta = \mathcal{O}(10^{-3}).$ Finally, our main result shows that under a suitable scaling, the logarithmic trading price process converges to a standard Brownian motion as $ \delta \to 0$. We also study the limit of the modified market ask and bid price processes $(A_\delta, B_\delta)$ as $\delta\to0$, which is exactly the original bouncing GBMs $(A, B)$. Using these asymptotics, we derive a simple and effective prediction formulas for trading prices.

It is interesting to see that we get an asymptotic GBM model for the trading prices in the limit. The GBM model captures the intuition that the rates of returns over non-overlapping intervals are independent of each other, and has been extensively used to model stock prices since the breakthrough made by Black and Scholes \cite{bs73} and Merton \cite{m76}. Another interesting observation is the logarithmic returns between consecutive trading times are NIG distributed. In fact, empirical studies show that logarithmic returns of assets can be fitted very well by NIG distributions (see \cite{bn95, bn96, r97}) and Barndorff-Nielsen proposed NIG models in \cite{bn98}. Thus our model of bouncing GBMs provides another justification for the GBM model of trading prices.

The rest of the paper is organized as follows. In Section \ref{sec:MF}, we introduce our model of bouncing GBMs in details. In Section \ref{sec:ttp} we construct the modified market ask and bid processes and the price-transaction process.  All the main results about the distributions of transaction times and prices, and the limiting behaviors are summarized in Section \ref{results}. In Section \ref{par}, the estimators of the model parameters are derived using the method of moments. In Section
\ref{numerical}, we use asymptotic GBM model obtained in Section \ref{results} for trading prices, from which we derive a simple and effective forecasting formula. We also apply the formula to real data, and show that the estimated $\delta$ parameter is indeed very small, and hence the asymptotic results are applicable, and work very well over short time horizons. Finally, all proofs are given in Appendix.

\section{Market ask and bid prices} \label{sec:MF}

We consider a trading system, where buyers and sellers arrive with specific prices. Recall that the market bid price is the largest price at which buyers are willing to buy, and the market ask price is the smallest price at which sellers are willing to sell. The market ask price cannot be less than the market bid price, and a trade occurs when the market bid and ask prices are matched. We will model the market bid and ask prices as {\em bouncing GBMs}, which are defined as exponentials of mutually reflected Brownian motions (BMs). More precisely, let $A(t)$ and $B(t)$ denote the market ask and bid prices at time $t\ge 0$, and assume that $A(0)\ge B(0)$. For $t\ge 0$, define
\begin{align}
X_{a}(t) &= \ln A(0) + \mu_a t + \sigma_a W_{a}(t), \label{gbm1}\\
X_{b}(t) &=\ln B(0) +\mu_b t + \sigma_b W_{b}(t),\label{gbm2}
\end{align}
where $W_{a}, W_{b}$ are independent standard BMs independent of $A(0)$ and $B(0)$, and $\mu_a, \mu_b$ and $\sigma_a, \sigma_b$ are the drift and variance parameters. We assume that $\mu_a < \mu_b.$ We first define a pair of mutually reflected BMs $(Y_a, Y_b)$ as follows. For $t\ge 0$, define
\begin{align}
Y_a(t) &= X_a(t) + \frac{1}{2}L(t), \label{mrbm-1}\\
Y_b(t) &=  X_b(t) - \frac{1}{2}L(t),\label{mrbm-2}
\end{align}
where $\{L(t), t\ge 0\}$ is the unique continuous nondecreasing process such that
\begin{itemize}
\item[\rm (i)] $L(0) =0$;
\item[\rm (ii)] $Y_a(t) - Y_b(t) \ge 0$ for all $t\ge 0$;
\item[\rm (iii)]$L(t)$ can increase only when $Y_a(t) -Y_b(t) =0$, i.e.,
\[
\int_0^\infty 1_{\{Y_a(t) - Y_b(t) >0\}} dL(t) = 0.
\]
\end{itemize}
 The existence and uniqueness of $\{L(t), t\ge 0\}$ are from Skorohod lemma (see \cite[Lemma 3.6.14]{ks91}). In fact, $L(t)$ has the following explicit formula
 \begin{align}\label{pushing}
 L(t) = \sup_{0\le s \le t} \left( X_a(s) - X_b(s)\right)^-,
 \end{align}
 where for $a\in\mathbb{R}, a^- = \max\{-a, 0\}.$ Rougly speaking, the processes $Y_a(t)$ and $Y_b(t)$ behave like two independent BMs when $Y_a(t) > Y_b(t)$, and whenever they meet, the process $Y_a(t)$ will be pushed up, while $Y_b(t)$ will be pushed down, to make $Y_a(t) \ge Y_b(t)$ for all $t\ge 0$. Here we assume the pushing effect for $Y_a(t)$ and $Y_b(t)$ are the same, and thus we have $\frac{1}{2}$ before the regulator process $L(t)$ in both \eqref{mrbm-1} and \eqref{mrbm-2}.

Finally, $A(t)$ and $B(t)$ are defined as
\begin{align}
A(t) & = e^{Y_a(t)}, \label{ask}\\
B(t) & = e^{Y_b(t)}. \label{bid}
\end{align}
Thus $A(t)$ and $B(t)$ behave like two independent GBMs when $A(t) > B(t)$, and whenever they become equal, they will be pushed away from each other such that $A(t)\ge B(t)$ for all $t\ge 0$.

One important quantity is the ratio $\frac{A(t)}{B(t)}$, which could reflect the ask-bid spread. We note that
\[
\frac{A(t)}{B(t)} = e^{Y_a(t) - Y_b(t)} = e^{X_a(t) - X_b(t) + L(t)}, \ t\ge 0.
\]
From \eqref{pushing}, $\{Y_a(t)-Y_b(t), t\ge 0\}$ is a reflected Brownian motion (RBM). It is well known that a RBM $\{R(t), t\ge 0\}$ with mean $\mu$ and variance $\sigma^2$ has the following transient cumulative distribution function (CDF) (see Section 1.8 in \cite{harrison85}). For $x, y \in [0,\infty),$
\begin{align}\label{rbm_disn}
\Pr(R(t) \le y| R(0) =x) = 1 - \Phi\left( \frac{-y+x+\mu t}{\sigma\sqrt{t}}\right) - e^{2\mu y/\sigma^2} \Phi\left( \frac{-y-x-\mu t}{\sigma\sqrt{t}}\right),
\end{align}
where $\Phi(\cdot)$ is the CDF of the standard normal distribution. Thus for $t\ge 0$, the ratio $\frac{A(t)}{B(t)}$ has the following CDF. Assuming $A(0)$ and $B(0)$ are deterministic constants, for $y\ge 1$,
\begin{align*}
 \Pr\left(\frac{A(t)}{B(t)} \le y\right) & = 1 - \Phi\left( \frac{-\ln[y]+\ln[A(0)/B(0)]+(\mu_a -\mu_b) t}{\sqrt{(\sigma^2_a+\sigma^2_b)t}}\right) \\
& \quad- y^{-\frac{2(\mu_b-\mu_a)}{\sigma_a^2+\sigma^2_b}}  \Phi\left( \frac{-\ln[y]-\ln[A(0)/B(0)]-(\mu_a-\mu_b) t}{\sqrt{(\sigma_a^2+\sigma_b^2)t}}\right).
\end{align*}
Consequently, under the condition that $\mu_a < \mu_b$, the stationary distribution of $\frac{A}{B}$ is power-law distributed with density function
\begin{align}\label{power-law}
\frac{2(\mu_b-\mu_a)}{\sigma_a^2+\sigma^2_b} y^{-1-\frac{2(\mu_b-\mu_a)}{\sigma_a^2+\sigma^2_b}}, \ \ y\ge 1.
\end{align}
It is interesting to see that only stationary moments of order less than $\frac{2(\mu_b-\mu_a)}{\sigma_a^2+\sigma^2_b}$ are finite.
For finite $t$, a simple description of the $k$-th moment of $\frac{A(t)}{B(t)}$ with $A(0)=B(0)$ is presented in the following lemma, the proof of which is provided in Appendix.
\begin{lemma}\label{ratiodisn} Assume $A(0)=B(0)$. Then for $k \in \mathbb{N}$,
\begin{align}\label{ratio-moment}
E\left[\left(\frac{A(t)}{B(t)}\right)^k\right] = 1 + \frac{k(\sigma_a^2+\sigma_b^2)}{\mu_b-\mu_a} \int_0^\infty \exp\left\{\frac{k(\sigma^2_a + \sigma^2_b)x}{\mu_b-\mu_a} - 2x \right \}F(t; x, 0) dx,
\end{align}
where \[
F(t; x, 0) = \Phi\left(\frac{t-x}{\sqrt{t}}\right) + e^{2x} \Phi\left(\frac{-t-x}{\sqrt{t}}\right), \ x\ge 0,
\]
is the CDF of the first-passage-time of $\{Y_a(t)-Y_b(t), t\ge 0\}$ from $x$ to $0$.
\end{lemma}
Note that $\lim_{t\to\infty} F(t; x, 0) = 1$, and so $\lim_{t\to\infty} E[(\frac{A(t)}{B(t)})^k]$ is finite only when $k < \frac{2(\mu_b-\mu_a)}{\sigma_a^2+\sigma^2_b}$. This result is consistent with the moments of the power law distribution in \eqref{power-law}, and indeed one can easily check that when $k < \frac{2(\mu_b-\mu_a)}{\sigma_a^2+\sigma^2_b}$,
\begin{align}\label{interchange}
\lim_{t\to\infty}E\left[\left(\frac{A(t)}{B(t)}\right)^k\right] = E\left[\frac{A(\infty)}{B(\infty)}\right],
\end{align}
where $\frac{A(\infty)}{B(\infty)}$ is a random variable with density function \eqref{power-law}.
Other performance analysis can be done by computing the joint distribution of $(A(t), B(t))$. However, it is nontrivial to obtain a simple description of the transient behavior of $(A(t), B(t)).$ Thus we would like to investigate such problems in a separate paper.

 \section{Trading times and prices} \label{sec:ttp}
Assuming that $A(0)>B(0)$, the first trading time is defined to be the first time that the market ask and bid prices become equal, and we would like to define the $n$th trading time to be the $n$th time the two prices become equal. However, the zero set $\{t\ge 0: A(t) - B(t) =0\}$ is uncountably infinite, and we cannot define the $n$th trading time as conveniently as the first one. Also note that in practice every time the market ask and bid prices become equal, they will separate from each other by at least one cent. Thus we consider the following modified market ask and bid price processes $A_\delta$ and $B_{\delta}$, where the positive constant $\delta$ represents the tick size.  We then use $A_\delta$ and $B_{\delta}$ to define the trading times and trading prices. More precisely, let $\delta$ be a strictly positive constant, and recall that $A(0)$ and $B(0)$ are the initial values of the market ask and bid price processes $A$ and $B$, and $X_a$ and $X_b$ are two independent BMs defined in \eqref{gbm1} and \eqref{gbm2}. For $n\ge 1$, define the following stopping times: $T_{\delta,0} = 0$, and
\begin{align}\label{tradingtimes}
T_{\delta,n}  = \inf\left\{t\ge 0:  X_a(t) - X_b(t) = {-2(n-1)\delta}\right\}.
\end{align}
Then $T_{\delta,n}\ge T_{\delta,n-1}$ and $T_{\delta,n}\to \infty$ almost surely as $n\to\infty$. We next define the modified market ask and bid price processes. For $t\ge 0,$
\begin{align}
A_\delta(t) & = \exp\left\{ X_a(t) + \sum_{n=1}^\infty (n-1)\delta 1_{\{t\in [T_{\delta,n-1}, T_{\delta,n})\}}\right\}, \\
B_\delta(t) & = \exp\left\{ X_b(t) - \sum_{n=1}^\infty (n-1)\delta 1_{\{t\in [T_{\delta,n-1}, T_{\delta,n})\}}\right\}.
\end{align}
Thus the first trade occurs at $T_{\delta,1}$, which is the first time the modified market ask and bid prices become equal, and the first trading price is defined as
\begin{align*}
P_{\delta,1} = A_\delta(T_{\delta,1}-)=B_\delta(T_{\delta,1}-) = e^{X_a(T_{\delta,1})} = e^{X_b(T_{\delta,1})}.
\end{align*}
(Note that $T_{\delta, 1}$ and $P_{\delta,1}$ don't depend on $\delta$ if the initials $A(0)$ and $B(0)$ are independent of $\delta$.) Right after the first trade occurs, the market ask and bid prices will separate in the following way.
\begin{align*}
A_\delta(T_{\delta,1}) = P_{\delta,1} e^{\delta} > P_{\delta,1}, \ \ B_\delta(T_{\delta,1}) = P_{\delta,1} e^{-\delta} < P_{\delta,1}.
\end{align*}
Starting from $T_{\delta, 1}$, the processes $A_\delta$ and $B_\delta$ evolves as two independent GMB's with initial values $P_{\delta,1} e^{\delta}$ and $P_{\delta,1} e^{-\delta}$ until they meet again at $T_{\delta, 2}$.
Recursively, for $n\ge 1$, the stopping time $T_{\delta,n}$ will be the $n$th meeting time of $A_\delta$ and $B_\delta$, and the $n$th trading price is defined as
\begin{align}
P_{\delta,n} = A_\delta(T_{\delta,n}-)=B_\delta(T_{\delta,n}-),
\end{align}
and the modified  market ask and bid prices at $T_{\delta,n}$ move to
\begin{align}
A_\delta(T_{\delta,n}) = P_{\delta,n}e^\delta > P_{\delta,n}, \ \ B_\delta(T_{\delta,n}) = P_{\delta,n}e^{-\delta} < P_{\delta,n}.
 \end{align}
Right after $T_{\delta,n}$, the processes $A_\delta$ and $B_\delta$ evolve as two independent GBMs with initials $P_{\delta,n}e^\delta$ and $P_{\delta,n}e^{-\delta}$ until they meet again at $T_{\delta, n+1}.$
 The dynamics of the market ask and bid prices is shown in Figure \ref{pricemov}.
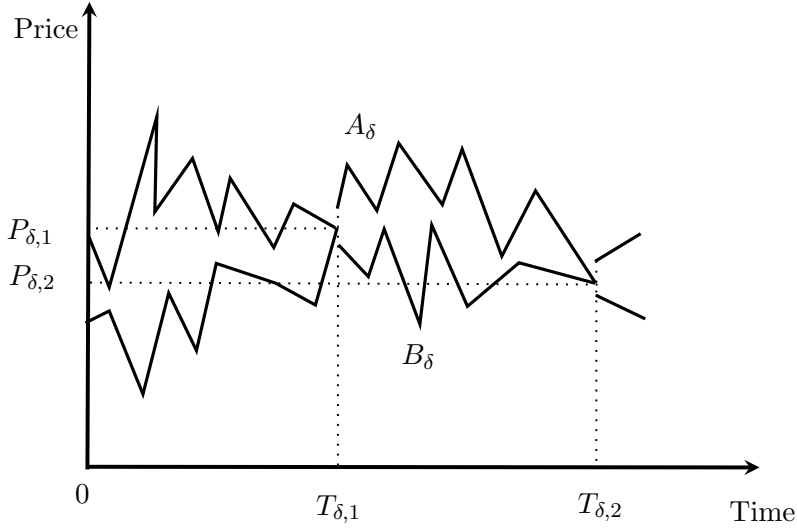
\begin{figure}
\centering
\ifx\du\undefined
  \newlength{\du}
\fi
\setlength{\du}{15\unitlength}
\begin{tikzpicture}
\pgftransformxscale{1.000000}
\pgftransformyscale{-1.000000}
\definecolor{dialinecolor}{rgb}{0.000000, 0.000000, 0.000000}
\pgfsetstrokecolor{dialinecolor}
\definecolor{dialinecolor}{rgb}{1.000000, 1.000000, 1.000000}
\pgfsetfillcolor{dialinecolor}
\pgfsetlinewidth{0.100000\du}
\pgfsetdash{}{0pt}
\pgfsetdash{}{0pt}
\pgfsetbuttcap
{
\definecolor{dialinecolor}{rgb}{0.000000, 0.000000, 0.000000}
\pgfsetfillcolor{dialinecolor}
\pgfsetarrowsend{stealth}
\definecolor{dialinecolor}{rgb}{0.000000, 0.000000, 0.000000}
\pgfsetstrokecolor{dialinecolor}
\draw (18.933333\du,15.232510\du)--(35.900000\du,15.249177\du);
}
\pgfsetlinewidth{0.100000\du}
\pgfsetdash{}{0pt}
\pgfsetdash{}{0pt}
\pgfsetbuttcap
{
\definecolor{dialinecolor}{rgb}{0.000000, 0.000000, 0.000000}
\pgfsetfillcolor{dialinecolor}
\pgfsetarrowsend{stealth}
\definecolor{dialinecolor}{rgb}{0.000000, 0.000000, 0.000000}
\pgfsetstrokecolor{dialinecolor}
\draw (18.950000\du,15.300000\du)--(19.000000\du,3.500000\du);
}
\definecolor{dialinecolor}{rgb}{0.000000, 0.000000, 0.000000}
\pgfsetstrokecolor{dialinecolor}
\node[anchor=west] at (18.350000\du,15.900000\du){0};
\definecolor{dialinecolor}{rgb}{0.000000, 0.000000, 0.000000}
\pgfsetstrokecolor{dialinecolor}
\node[anchor=west] at (34.850000\du,16.350000\du){Time};
\definecolor{dialinecolor}{rgb}{0.000000, 0.000000, 0.000000}
\pgfsetstrokecolor{dialinecolor}
\node[anchor=west] at (16.800000\du,4.150000\du){Price};
\pgfsetlinewidth{0.080000\du}
\pgfsetdash{}{0pt}
\pgfsetdash{}{0pt}
\pgfsetmiterjoin
\pgfsetbuttcap
{
\definecolor{dialinecolor}{rgb}{0.000000, 0.000000, 0.000000}
\pgfsetfillcolor{dialinecolor}
{\pgfsetcornersarced{\pgfpoint{0.000000\du}{0.000000\du}}\definecolor{dialinecolor}{rgb}{0.000000, 0.000000, 0.000000}
\pgfsetstrokecolor{dialinecolor}
\draw (18.975000\du,9.400000\du)--(19.500000\du,10.700000\du)--(20.700000\du,6.400000\du)--(20.650000\du,8.800000\du)--(21.600000\du,7.450000\du)--(22.250000\du,9.300000\du)--(22.550000\du,7.950000\du)--(23.650000\du,9.700000\du)--(24.150000\du,8.600000\du)--(25.200000\du,9.200000\du);
}}
\pgfsetlinewidth{0.080000\du}
\pgfsetdash{}{0pt}
\pgfsetdash{}{0pt}
\pgfsetmiterjoin
\pgfsetbuttcap
{
\definecolor{dialinecolor}{rgb}{0.000000, 0.000000, 0.000000}
\pgfsetfillcolor{dialinecolor}
{\pgfsetcornersarced{\pgfpoint{0.000000\du}{0.000000\du}}\definecolor{dialinecolor}{rgb}{0.000000, 0.000000, 0.000000}
\pgfsetstrokecolor{dialinecolor}
\draw (18.900000\du,11.600000\du)--(19.500000\du,11.300000\du)--(20.350000\du,13.400000\du)--(21.000000\du,10.850000\du)--(21.700000\du,12.300000\du)--(22.200000\du,10.100000\du)--(23.700000\du,10.600000\du)--(24.700000\du,11.150000\du)--(25.250000\du,9.200000\du);
}}
\pgfsetlinewidth{0.080000\du}
\pgfsetdash{}{0pt}
\pgfsetdash{}{0pt}
\pgfsetmiterjoin
\pgfsetbuttcap
{
\definecolor{dialinecolor}{rgb}{0.000000, 0.000000, 0.000000}
\pgfsetfillcolor{dialinecolor}
{\pgfsetcornersarced{\pgfpoint{0.000000\du}{0.000000\du}}\definecolor{dialinecolor}{rgb}{0.000000, 0.000000, 0.000000}
\pgfsetstrokecolor{dialinecolor}
\draw (25.250000\du,8.716667\du)--(25.500000\du,7.616667\du)--(26.250000\du,8.766667\du)--(26.800000\du,7.066667\du)--(27.900000\du,8.616667\du)--(28.400000\du,7.216667\du)--(29.400000\du,9.916667\du)--(30.250000\du,8.266667\du)--(31.766667\du,10.600000\du);
}}
\pgfsetlinewidth{0.080000\du}
\pgfsetdash{}{0pt}
\pgfsetdash{}{0pt}
\pgfsetmiterjoin
\pgfsetbuttcap
{
\definecolor{dialinecolor}{rgb}{0.000000, 0.000000, 0.000000}
\pgfsetfillcolor{dialinecolor}
{\pgfsetcornersarced{\pgfpoint{0.000000\du}{0.000000\du}}\definecolor{dialinecolor}{rgb}{0.000000, 0.000000, 0.000000}
\pgfsetstrokecolor{dialinecolor}
\draw (25.283333\du,9.633333\du)--(26.033333\du,10.433333\du)--(26.433333\du,9.233333\du)--(27.333333\du,11.633333\du)--(27.633333\du,9.133333\du)--(28.533333\du,11.183333\du)--(29.833333\du,10.083333\du)--(31.750000\du,10.600000\du);
}}
\pgfsetlinewidth{0.080000\du}
\pgfsetdash{}{0pt}
\pgfsetdash{}{0pt}
\pgfsetmiterjoin
\pgfsetbuttcap
{
\definecolor{dialinecolor}{rgb}{0.000000, 0.000000, 0.000000}
\pgfsetfillcolor{dialinecolor}
{\pgfsetcornersarced{\pgfpoint{0.000000\du}{0.000000\du}}\definecolor{dialinecolor}{rgb}{0.000000, 0.000000, 0.000000}
\pgfsetstrokecolor{dialinecolor}
\draw (31.750000\du,10.050000\du)--(32.900000\du,9.350000\du);
}}
\pgfsetlinewidth{0.080000\du}
\pgfsetdash{}{0pt}
\pgfsetdash{}{0pt}
\pgfsetmiterjoin
\pgfsetbuttcap
{
\definecolor{dialinecolor}{rgb}{0.000000, 0.000000, 0.000000}
\pgfsetfillcolor{dialinecolor}
{\pgfsetcornersarced{\pgfpoint{0.000000\du}{0.000000\du}}\definecolor{dialinecolor}{rgb}{0.000000, 0.000000, 0.000000}
\pgfsetstrokecolor{dialinecolor}
\draw (31.766667\du,10.900000\du)--(33.016667\du,11.500000\du);
}}
\pgfsetlinewidth{0.050000\du}
\pgfsetdash{{\pgflinewidth}{0.200000\du}}{0cm}
\pgfsetdash{{\pgflinewidth}{0.200000\du}}{0cm}
\pgfsetbuttcap
{
\definecolor{dialinecolor}{rgb}{0.000000, 0.000000, 0.000000}
\pgfsetfillcolor{dialinecolor}
\definecolor{dialinecolor}{rgb}{0.000000, 0.000000, 0.000000}
\pgfsetstrokecolor{dialinecolor}
\draw (25.266667\du,8.650000\du)--(25.266667\du,15.300000\du);
}
\pgfsetlinewidth{0.050000\du}
\pgfsetdash{{\pgflinewidth}{0.200000\du}}{0cm}
\pgfsetdash{{\pgflinewidth}{0.200000\du}}{0cm}
\pgfsetbuttcap
{
\definecolor{dialinecolor}{rgb}{0.000000, 0.000000, 0.000000}
\pgfsetfillcolor{dialinecolor}
\definecolor{dialinecolor}{rgb}{0.000000, 0.000000, 0.000000}
\pgfsetstrokecolor{dialinecolor}
\draw (31.783333\du,10.033333\du)--(31.783333\du,15.215843\du);
}
\pgfsetlinewidth{0.050000\du}
\pgfsetdash{{\pgflinewidth}{0.200000\du}}{0cm}
\pgfsetdash{{\pgflinewidth}{0.200000\du}}{0cm}
\pgfsetbuttcap
{
\definecolor{dialinecolor}{rgb}{0.000000, 0.000000, 0.000000}
\pgfsetfillcolor{dialinecolor}
\definecolor{dialinecolor}{rgb}{0.000000, 0.000000, 0.000000}
\pgfsetstrokecolor{dialinecolor}
\draw (25.216667\du,9.216667\du)--(18.966667\du,9.216667\du);
}
\pgfsetlinewidth{0.050000\du}
\pgfsetdash{{\pgflinewidth}{0.200000\du}}{0cm}
\pgfsetdash{{\pgflinewidth}{0.200000\du}}{0cm}
\pgfsetbuttcap
{
\definecolor{dialinecolor}{rgb}{0.000000, 0.000000, 0.000000}
\pgfsetfillcolor{dialinecolor}
\definecolor{dialinecolor}{rgb}{0.000000, 0.000000, 0.000000}
\pgfsetstrokecolor{dialinecolor}
\draw (31.816667\du,10.633333\du)--(18.983333\du,10.583333\du);
}
\definecolor{dialinecolor}{rgb}{0.000000, 0.000000, 0.000000}
\pgfsetstrokecolor{dialinecolor}
\node[anchor=west] at (24.450000\du,16.300000\du){$T_{\delta,1}$};
\definecolor{dialinecolor}{rgb}{0.000000, 0.000000, 0.000000}
\pgfsetstrokecolor{dialinecolor}
\node[anchor=west] at (31.050000\du,16.250000\du){$T_{\delta,2}$};
\definecolor{dialinecolor}{rgb}{0.000000, 0.000000, 0.000000}
\pgfsetstrokecolor{dialinecolor}
\node[anchor=west] at (17.250000\du,10.700000\du){};
\definecolor{dialinecolor}{rgb}{0.000000, 0.000000, 0.000000}
\pgfsetstrokecolor{dialinecolor}
\node[anchor=west] at (16.700000\du,10.500000\du){$P_{\delta,2}$};
\definecolor{dialinecolor}{rgb}{0.000000, 0.000000, 0.000000}
\pgfsetstrokecolor{dialinecolor}
\node[anchor=west] at (16.650000\du,9.350000\du){$P_{\delta,1}$};
\definecolor{dialinecolor}{rgb}{0.000000, 0.000000, 0.000000}
\pgfsetstrokecolor{dialinecolor}
\node[anchor=west] at (25.150000\du,6.650000\du){$A_\delta$};
\definecolor{dialinecolor}{rgb}{0.000000, 0.000000, 0.000000}
\pgfsetstrokecolor{dialinecolor}
\node[anchor=west] at (26.600000\du,12.450000\du){$B_\delta$};
\end{tikzpicture}
\caption{Dynamics of the modified market ask and bid prices $(A_\delta, B_\delta)$.}\label{pricemov}
\end{figure}

The relationship between the modified market ask and bid prices $(A_\delta, B_\delta)$ and the original market ask and bid prices $(A, B)$ is summarized in the following proposition. Its proof can be found in Appendix. In particular, it shows that $\{T_{\delta,n}\}_{n\in\mathbb{N}}$ are also the meeting times of the original price processes $A$ and $B$, and that $(A_\delta(t), B_\delta(t))$ converges to $(A(t), B(t))$ almost surely and uniformly on compact sets of $[0,\infty)$ as $\delta \to 0$.
\begin{proposition}\label{convergence-MRBM} \hfill
\begin{itemize}
\item[\rm (i)] For $\delta>0$ and $n\in\mathbb{N}$, 
\[
A(T_{\delta,n}) = B(T_{\delta,n}).
\]
\item[\rm (ii)]For $\delta>0$ and $t\ge 0,$ 
\[
A_\delta(t) \ge A(t), \ \mbox{and} \ B_\delta(t) \le B(t).
\]
\item[\rm (iii)] For $t\ge 0,$
\[
\sup_{0\le s \le t} \frac{A_\delta(t)}{A(t)} \to 1, \ \mbox{and} \ \sup_{0\le s \le t} \frac{B(t)}{B_\delta(t)}\to 1, \ \ \mbox{almost surely as $\delta \to 0.$}
\]
\end{itemize}
\end{proposition}

For convenience, we denote
\begin{align*}
& U_{\delta,n+1} = \ln (P_{\delta,n+1}/P_{\delta,n}),\;\;V_{\delta, n+1} = T_{\delta,n+1} - T_{\delta,n}, \ n\ge 1, \\
& U_{\delta,1} = \ln P_{\delta,1}, \ \ V_{\delta,1} = T_{\delta,1}.
\end{align*}
We are interested in the evolution of the trading prices. Define for $t\ge 0$,
\begin{align}
N_\delta(t) =\max\{n\geq 0: T_{\delta,n} \leq t\},
\end{align}
which gives the number of trades up to time $t$. Now the latest trading price can be formulated as
\begin{align}
P_\delta(t) = P_{\delta, N_\delta(t)}, \ \mbox{for} \ t\ge T_{\delta,1}.
\end{align}
For $t\ge T_{\delta,1}$, let $Z_\delta(t)= \ln(P_\delta(t))$, and so $Z_\delta(t) = \sum_{n=1}^{N_\delta(t)} U_{\delta,n}.$
When $0\le t < T_{\delta,1}$, we simply let $Z_\delta(t) =0$. Thus we have
\begin{align}
Z_\delta(t) = \sum_{n=1}^{N_\delta(t)} U_{\delta,n},
\end{align}
with the convention that $\sum_{n=1}^0 U_{\delta,n}= 0.$
We will see in Lemma \ref{rrp} that $\{Z_\delta(t), t\ge0\}$ is a renewal reward process. Our goal is to establish a scaling limit theorem for $Z$ as $\delta \to 0$, and develop an asymptotic model for real financial data.

\section{Main results}\label{results}
We present our main results in this section. In particular, it is shown that $(U_{\delta,n}, V_{\delta,n}), n\ge 2,$ are i.i.d. random variables (see Lemma \ref{distn-1}), and $U_{\delta,n}$ follows a NIG distribution and $V_{\delta,n}$ is IG distributed (see Corollary \ref{marginal}). Using these results, it is clear that $\{Z_\delta(t), t\ge 0\}$ is a renewal reward process, and the scaling limit theorem is established in Theorem \ref{th:conv}. All the proofs are provided in Appendix.

\subsection{Distribution of $(U_{\delta, n}, V_{\delta, n})$}

We derive the joint distribution of $(U_{\delta,n}, V_{\delta,n})$ for each $n\ge 1$ in the following lemma. Note that $(U_{\delta,1}, V_{\delta,1})$ doesn't depend on $\delta$ if the intial values $A(0)$ and $B(0)$ are independent of $\delta$.

\begin{lemma}\label{distn-1} \hfill
\begin{itemize}
\item[\rm (i)] Assume $A(0) = e^\alpha, B(0) = e^\beta$, and $\alpha > \beta$. For $t\ge 0$ and $x\in \mathbb{R},$ we have
\begin{align}
& \Pr(U_{\delta,1}\in dx, V_{\delta,1}\in dt) = \nonumber\\
& \!\!\! \frac{\alpha-\beta}{2\pi t^2 \sigma_a \sigma_b} \exp\left\{-\frac{\left[\frac{\sigma_b}{\sigma_a}(x-\alpha-\mu_at)+\frac{\sigma_a}{\sigma_b}(x-\beta-\mu_bt)\right]^2+[\alpha-\beta - (\mu_b -\mu_a)t]^2}{2(\sigma_a^2 + \sigma_b^2) t}\right\} dx dt. \label{distribution}
\end{align}
In particular, $V_{\delta,1}$ follows IG distribution with the following density function
\begin{align}
\Pr(V_{\delta,1}\in dt) =  \frac{\alpha-\beta}{\sqrt{2\pi(\sigma_a^2 + \sigma_b^2) t^3} } \exp\left\{-\frac{[\alpha-\beta - (\mu_b -\mu_a)t]^2}{2(\sigma_a^2 + \sigma_b^2) t}\right\} dt,\label{meeting-time}
\end{align}
and given $V_{\delta,1} = t$, $U_{\delta,1}$ follows normal distribution with mean $\frac{\sigma^2_b(\alpha+\mu_a t)+\sigma^2_a(\beta+\mu_b t)}{\sigma^2_a + \sigma^2_b}$ and variance $\frac{\sigma_a^2\sigma^2_b t}{\sigma^2_a + \sigma^2_b}$.
\item[\rm (ii)] The sequence $(U_{\delta,n}, V_{\delta,n})_{n\ge 2}$ is an i.i.d. sequence, which is independent of $(U_{\delta,1}, V_{\delta,1})$ and has the same distribution as in (i) with $\alpha = \delta$ and $\beta = -\delta.$
\end{itemize}
\end{lemma}

To derive the marginal distributions of $U_n, n\ge 1$, we introduce the following definitions of IG and NIG distributions (see \cite{s93}).

\begin{definition}
\begin{itemize}
\item[\rm (i)] An inverse  Gaussian (IG) distribution with parameters $a_1$ and $a_2$ has density function
\[
f(x; a_1, a_2) = \frac{a_1}{\sqrt{2\pi x^3}} \exp\left\{-\frac{(a_1 -a_2 x)^2}{2 x}\right\}, \ x>0,
\]
which is usually denoted by $\IG(a_1, a_2).$
\item[\rm (ii)] A random variable $X$ follows a normal inverse Gaussian (NIG) distribution with parameters $\bar\alpha, \bar\beta, \bar\mu, \bar\delta$ with notation $\NIG(\bar\alpha, \bar\beta, \bar\mu, \bar\delta)$ if
\[
 Y|X=x \sim N(\bar\mu + \bar\beta x, x), \ \mbox{and} \ X\sim \IG(\bar\delta, \sqrt{\bar\alpha^2 - \bar\beta^2}).
\]
The density function of $Y$ is given as
\[
f(y;\bar\alpha, \bar\beta, \bar\mu, \bar\delta) = \frac{\bar\alpha}{\pi\bar\delta} \exp\left\{\sqrt{\bar\alpha^2 - \bar\beta^2} + \frac{\bar\beta}{\bar\delta}(y -\bar\mu)\right\} \frac{K_1\left(\bar\alpha\sqrt{1 + (\frac{y-\bar\mu}{\bar\delta})^2}\right)}{\sqrt{1 + (\frac{y-\bar\mu}{\bar\delta})^2}},
\]
where $K_1(z) = \frac{1}{2}\int_0^\infty e^{-z(t+t^{-1})/2} dt$ is the modified Bessel function of the third kind with index $1$.
\end{itemize}
\end{definition}

Using the above definitions, we have the following conclusion on the marginal distributions of $(U_n, V_n), n\ge 1$.
\begin{corollary}\label{marginal}
\begin{itemize}
\item[\rm (i)] Assume $A(0) = e^\alpha$, $B(0) = e^\beta$, and $\alpha>\beta$. Then
\begin{align*}
V_{\delta,1} \sim \IG\left(\frac{\alpha-\beta}{\sqrt{\sigma^2_a +\sigma^2_b}}, \ \frac{\mu_b-\mu_a}{\sqrt{\sigma^2_a + \sigma^2_b}} \right),
\end{align*}
and
\begin{align*}
& U_{\delta,1}  \sim \NIG\left(\frac{\sqrt{(\sigma^2_a + \sigma^2_b)(\mu_a^2\sigma_b^2+\mu_b^2\sigma_a^2)}}{\sigma_a^2\sigma_b^2}, \ \frac{\mu_a\sigma_b^2+\mu_b\sigma_a^2}{\sigma_a^2\sigma_b^2}, \ \frac{\alpha\sigma_b^2 +\beta\sigma_a^2}{\sigma_a^2+\sigma_b^2}, \ \frac{(\alpha-\beta)\sigma_a\sigma_b}{\sigma_a^2+\sigma_b^2} \right). \label{nig-1}
\end{align*}
\item[\rm (ii)] For $n\ge 2$, $V_{\delta,n}$ and $U_{\delta,n}$ follow the same IG and NIG distributions as in (i) with $\alpha = \delta$ and $\beta = -\delta.$
\end{itemize}
\end{corollary}

Let $(U_\delta,V_\delta)$ be a generic random variable with the same joint distribution as $(U_{\delta,n},V_{\delta,n}), n\ge 2$. Next we find the moment generating function of $(U_\delta,V_\delta)$, which will be used in the proof of Theorem \ref{th:conv} and Section \ref{par}.

\begin{lemma}\label{thm2_moments}
There exists $h>0$ such that the moment generating function of $(U_\delta, V_\delta)$ exists for $|(s,t)|\le h$, and is given by
\begin{equation}\label{mgf}
\phi_\delta(s,t) = E\left[ {\exp \{   sU_\delta + tV_\delta\}} \right] = \exp \{ [2{\theta}(s,t)-s]\delta \},
\end{equation}
where $\theta(s,t)$ is defined as follows.
\begin{align}\label{theta1}
{\theta(s,t)} &= \frac{{({\mu _b} - {\mu _a} + s\sigma _b^2) - \sqrt {{{({\mu _b} - {\mu _a} + s\sigma_b^2)}^2} - (\sigma _a^2 + \sigma _b^2)({s^2}\sigma _b^2 + 2t + 2s{\mu _b})} }}{{\sigma _a^2 + \sigma _b^2}}.
\end{align}
In particular, the first two moments of $(U_\delta,V_\delta)$ are as given below:
\begin{eqnarray*}
&&E(V_\delta) = \frac{2\delta }{{{\mu _b} - {\mu _a}}},\;\;E({U_\delta} ) = \frac{{\delta ({\mu _b} + {\mu _a})}}{{{\mu
_b} - {\mu _a}}},\\
&&\mbox{Var}(V_\delta) = \frac{{2(\sigma _a^2 + \sigma _b^2)\delta }}{{{{({\mu _b} - {\mu _a})}^3}}},\;\;\mbox{Var}(U_\delta) = \frac{{2(\mu _b^2\sigma _a^2 + \mu _a^2\sigma _b^2)\delta }}{{{{({\mu _b} - {\mu _a})}^3}}}, \\
&&\mbox{Cov}(U_\delta,V_\delta) = \frac{{2({\mu _b}\sigma _a^2 + {\mu _a}\sigma _b^2)\delta }}{{{{({\mu _b} - {\mu _a})}^3}}}.
\end{eqnarray*}
Furthermore, for $k, l\in \mathbb{N}\cup\{0\}$ and $k+l\ge 1$, there exists some constant $c_0$ such that
\begin{equation}\label{moments}
\frac{E(U^k_\delta V^l_\delta)}{\delta} \to c_0, \ \ \mbox{as $\delta\to0$.}
\end{equation}
\end{lemma}

\subsection{Asymptotics of $\{Z_\delta(t), t\ge 0\}$}
In this section we study the behaviors of the $\{Z_\delta(t), t \ge 0\}$ process as either $t\to\infty$ or $\delta \rightarrow 0$. First from Corollary \ref{marginal}, it is clear that for each $\delta$, $\{Z_\delta(t), t\geq 0\}$ is a renewal reward process, and we summarize it in the following lemma.

\begin{lemma}\label{rrp}
For $\delta>0$, $\{Z_\delta(t), t\geq 0\}$ is a renewal reward process and $\{P_\delta(t), t \ge 0\}$ is a semi-Markov process.
\end{lemma}

The next result from Brown and Solomon \cite{bs75} characterizes the asymptotic first and second moments of $Z_\delta(t)$ as $t\to\infty$, and is also helpful to identify the proper scaling in Theorem \ref{th:conv}.
\begin{theorem}[Brown and Solomon \cite{bs75}] \label{th:asy}
We have
\begin{equation} \label{eq:ezt}
E(Z_\delta(t)) = mt + O(1),
\end{equation}
where
\begin{equation}\label{eq:m}
m= \frac{1}{2}(\mu_a+\mu_b),
\end{equation}
and
\begin{equation} \label{eq:vzt}
\mbox{Var}(Z_\delta(t)) = s t + O(1),
\end{equation}
where
\begin{equation} \label{eq:s}
s= \frac{1}{4}(\sigma_a^2 + \sigma_b^2).
\end{equation}
Here $O(1)$ is a function that converges to a finite constant as $t \rightarrow \infty$.
\end{theorem}

The  main result is given in the following theorem. For $t\ge 0$, define
\[ \clz_\delta(t) = \frac{\delta Z_\delta(t/\delta) - mt}{\sqrt{s\delta}},\]
where $m$ and $s$ are as given in \eqref{eq:m} and \eqref{eq:s}.

\begin{theorem} \label{th:conv}
Assume that $E(\ln^2[A(0)/B(0)])<\infty$. Then the process $\clz_\delta$ converges weakly to a standard Brownian Motion as $\delta \to 0$.
\end{theorem}

\begin{remark}
We note that
\[
Z_\delta(t) =\sqrt{\frac{s}{\delta}} \hat Z_\delta(\delta t) + mt, \ t\ge 0.
\]
From Theorem \ref{th:conv}, for small $\delta$, we will use the following asymptotic model for logarithmic trading prices $Z(t)$ in Section \ref{numerical}:
\begin{align}\label{asymp-model}
\sqrt{\frac{s}{\delta}} B(\delta t) + mt,
\end{align}
where $\{B(t), t\ge 0\}$ is a standard Brownian motion. We note that \eqref{asymp-model} is normal distributed with mean $mt$ and variance $st$.
\end{remark}

\section{Parameter estimations} \label{par}
The process $\{P(t), t\geq 0\}$ is observable, while $\{(A(t), B(t)), t\geq 0\}$ may not be
publicly observable. The market ask and bid processes may be accessible to the brokers and dealers, but not to common traders. The question becomes how to find the parameters of $\{(A(t), B(t)), t\geq 0\}$ by observing $\{P(t), t\geq 0\}$. In this section we will estimate the parameters $\mu_a$, $\mu_b$, $\sigma_a$, $\sigma_b$ and $\delta$ using the method of moments.

Suppose that the sample data for the $i$th trading time $t_i$ and the $i$th trading
price $p_i$ are given for $i=1,2,\cdots,n$. Let
\[u_1 = \ln P_1, \ \ v_1 = t_1, \ \ \mbox{and} \ \ u_{i+1} = \ln (p_{i+1}/p_i),\;\;v_{i+1} = t_{i+1}-t_{i}, \ i\ge 1.\]
Then the sample data is given by $\{(u_i, v_i)\}_{i=1}^n$. Let
\begin{align*}
x_1 &= \sum\limits_{i = 1}^n \frac{{{v_i}}}{n},\;\;x_2 = \sum\limits_{i = 1}^n \frac{{u_i}}{n}, \;\; x_3 = \sum\limits_{i = 1}^n \frac{{v_i^2}}{n}, \ \ x_4 = \sum\limits_{i = 1}^n \frac{{u_i^2}}{n}, \ \ x_5 =\sum\limits_{i = 1}^n \frac{{v_i u_i}}{n}.
\end{align*}
We aim to derive explicit estimators of the five parameters $\mu_a$, $\mu_b$, $\sigma_a$, $\sigma_b$, $\delta$ using moment estimations. Define the estimators of $\mu_a$, $\mu_b$, $\sigma_a$, $\sigma_b$, $\delta$ as follows.
\begin{equation}\label{parameter-est}
\begin{aligned}
{\hat\mu _a}^n &= \frac{{{y_1} - {\sqrt{{y_1^2 - \frac{4({{y_1}{y_4} - {y_3}})}{{{y_2}}}}}}}}{2},\;\;{\hat\mu _b}^n = \frac{{{y_1} + {\sqrt{ {y_1^2 - \frac{4({{y_1}{y_4} - {y_3}})}{{{y_2}}}} }}}}{2}, \\
{\hat\sigma _a}^n &= \sqrt {({y_4} - {\hat\mu_a}^n{y_2})({\hat\mu_b}^n - {\hat\mu_a}^n)},\;\;{\hat\sigma _b}^n = \sqrt {({\hat\mu_b}^n{y_2} - {y_4})({\hat\mu_b}^n - {\hat\mu_a}^n)}, \\
\hat\delta^n  &= ({\hat\mu_b}^n - {\hat\mu_a}^n){x_1},
\end{aligned}
\end{equation}
where
\begin{align*}
y_1 = \frac{2x_2}{x_1},\;\;y_2 = \frac{x_3 - x_1^2}{x_1}, \;\;y_3 = \frac{x_4 - x_2^2}{x_1}, \;\;y_4 = \frac{x_5 - x_1 x_2}{x_1}.
\end{align*}
For convinence, denote $\Theta = (\mu_a, \mu_b, \sigma_a, \sigma_b, \delta)$ and $\hat\Theta^n = (\hat\mu_a^n, \hat\mu_b^n, \hat\sigma_a^n, \hat\sigma_b^n, \hat\delta^n).$
Let $\mathbf{g}: \mathbb{R}^5 \to \mathbb{R}^5$ be the differentiable function such that
\[
\hat\Theta^n = \mathbf{g}(x_1, x_2, x_3, x_4, x_5).
\]
Note that $\mathbf{g}$ can be uniquely determined by \eqref{parameter-est} and has an explicit expression.
\begin{lemma}\label{estimation1}
 The estimators $\hat\Theta^n$ is well defined, i.e.,
\begin{align}
 {y_1^2 - \frac{4({{y_1}{y_4} - {y_3}})}{{{y_2}}}} \geq 0, \ \ ({y_4} - {\hat\mu_a}^n{y_2})({\hat\mu_b}^n - {\hat\mu_a}^n) \geq 0, \ \ ({\hat\mu_b}^n{y_2} - {y_4})({\hat\mu_b}^n - {\hat\mu_a}^n) \geq 0, \label{well-define}
\end{align}
and as $n\to\infty,$
\begin{align}\label{consistent}
\hat\Theta^n \to \Theta, \ \mbox{almost surely.}
\end{align}
Furthermore, $\sqrt{n}(\hat\Theta^n - \Theta)$ converges weakly to a five dimensional normal distribution with zero mean and covariance matrix $\nabla \mathbf{g}(\Theta) \Sigma$, where $\Sigma$ is the covariance matrix of $(V_\delta, U_\delta, V^2_\delta, U^2_\delta, U_\delta V_\delta)$, and $\nabla \mathbf{g}$ is the gradient of $\mathbf{g}.$
\end{lemma}

\section{Numerical examples} \label{numerical}

In this section we apply our model to the real data, with an aim to forecast the trading price movement over a short period. We develop an asymptotic GBM model for trading prices as follows. Given the sample data $\{(u_i, v_i)\}_{i=1}^n$, we first estimate the parameters $\mu_a, \mu_b, \sigma_a, \sigma_b$ and $\delta$ as in \eqref{parameter-est}, and use the estimators $\hat\mu_a^n, \hat\mu_b^n, \hat\sigma_a^n$ and $\hat\sigma_b^n$ to compute $m$ and $s$ by substituting $\mu_a, \mu_b, \sigma_a, \sigma_b$ with $\hat\mu_a^n, \hat\mu_b^n, \hat\sigma_a^n,\hat\sigma_b^n$, respectively, in \eqref{eq:m} and \eqref{eq:s}. Typically, the estimator $\hat\delta^n$ is small (see Figures 6 - 9) and so from Theorem \ref{th:conv}, we approximate $Z(t)$ by a $N(st, mt)$ random variable. Hence the prediction formula for $\ln P(t) - \ln P(0)$ is
\[
 \frac{(\hat\mu_a^n + \hat\mu_b^n)t}{2},
\]
and the upper and lower bounds are chosen to be
\[
\frac{(\hat\mu_a^n + \hat\mu_b^n)t}{2} + \frac{3\sqrt{[(\hat\sigma^n_a)^2+(\hat\sigma^n_b)^2]t}}{2}, \ \ \frac{(\hat\mu_a^n + \hat\mu_b^n)t}{2}  - \frac{3\sqrt{[(\hat\sigma^n_a)^2+(\hat\sigma^n_b)^2]t}}{2}.
\]

We next apply the above formulas to real data. Here we select the stock SUSQ (Susquehanna Bancshares Inc). The data is chosen from 01/04/2010 9:30AM to 01/04/2010 4:00PM, including the trading prices and trading times. The unit of trading prices is dollars and the unit of the difference of consecutive trading times is seconds. We perform the back test to evaluate the performance of the prediction. To be precise,  we predict the logarithmic trading price at each trading time using the $10$-minute data $1$-minute before the trading time. For example, observing that there is a trade at 10:34:56, we then use the data from 10:23:56 to 10:33:56 to estimate the parameters and predict the logarithmic trading price at 10:34:56, and the last trading price during the time interval from 10:23:56 to 10:33:56 is regarded as $P(0)$.  At the same time we calculate the upper and lower bounds of the prediction at that trading time. We note that even though the drift and volatility parameters in the asymptotic model \eqref{asymp-model} is constant, the estimated parameters for predictions are actually time-varying. We compare this
predicted logarithmic trading prices with the real trading prices in Figure \ref{fig:1}. We do the similar prediction for each trading time but using the $10$-minute data $2$-, $5$-, $10$-minute before the trading time respectively. The comparisons are shown in Figures \ref{fig:2}-\ref{fig:4}. Define
\[
\mbox{Relative error (RE)} = \frac{\mbox{Real price - Predicted price}}{\mbox{Real price}}.
\]
For the predictions $1$-, $2$-, $5$-, $10$-minute into the future, the maximum absolute REs are $0.0055$, $0.0058$, $0.0080$, $0.0152$, respectively. We see that the prediction $1$-minute into the future provides very good forecasting, and the accuracy of the prediction deteriorates as we try to predict farther into the future, which is to be expected. We note that our asymptotic model is obtained when $\delta$ is small. We present the values of $\hat\delta^n$ for all four predictions in Figures 6 - 9, and observe that all values are $\mathcal{O}(10^{-3})$. Thus it is reasonable to use the asymptotic results in the regime $\delta \rightarrow 0.$

\begin{figure}[H]
\centering
  \includegraphics[height=7cm, width=13cm, keepaspectratio]{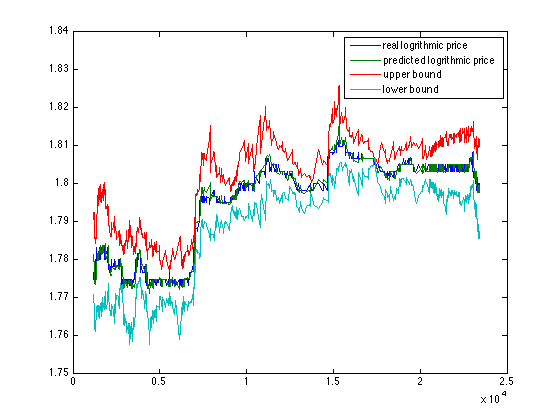}
  \caption{Predictions of trading prices using $10$-minute data $1$-minute before each trading time.} \label{fig:1}
\end{figure}

\begin{figure}[H]
\centering
  \includegraphics[height=7cm, width=13cm, keepaspectratio]{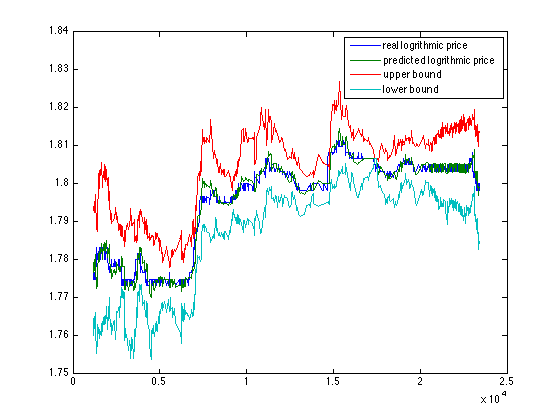}
  \caption{Predictions of trading prices using $10$-minute data 2-minute before each trading time.} \label{fig:2}
\end{figure}

\begin{figure}[H]
\centering
  \includegraphics[height=7cm, width=13cm, keepaspectratio]{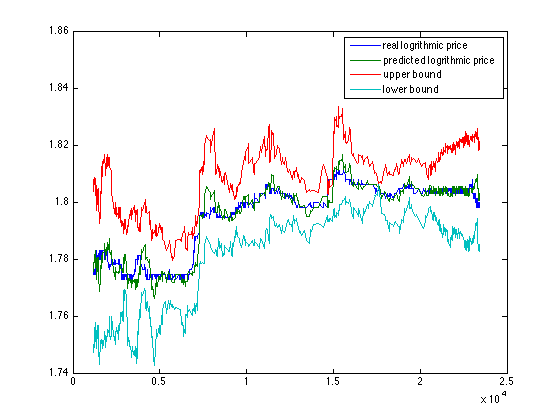}
  \caption{Predictions of trading prices using $10$-minute data $5$-minute before each trading time.} \label{fig:3}
\end{figure}

\begin{figure}[H]
\centering
  \includegraphics[height=7cm, width=13cm, keepaspectratio]{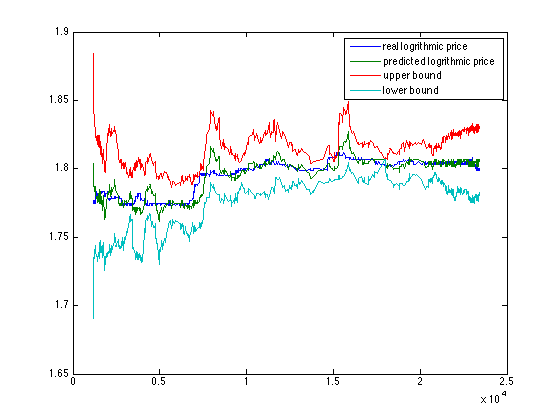}
  \caption{Predictions of trading prices using $10$-minute data $10$-minute before each trading time.} \label{fig:4}
\end{figure}

\begin{figure}[H]
\centering
  \includegraphics[width=12.5cm]{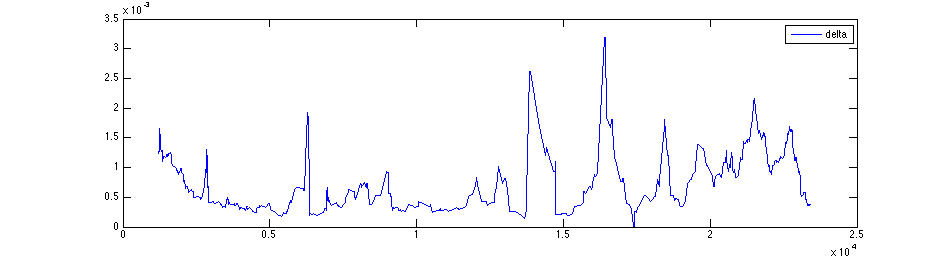}
  \caption{Values of $\hat\delta^n$ when using $10$-minute data $1$-minute before each trading time.}
\end{figure}

\begin{figure}[H]
\centering
  \includegraphics[width=12.5cm]{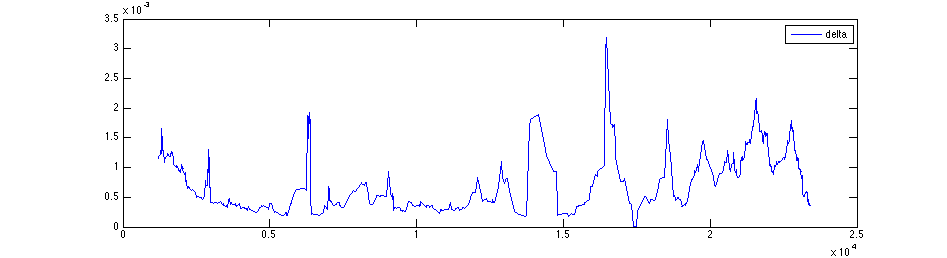}
  \caption{Values of $\hat\delta^n$ when using $10$-minute data 2-minute before each trading time.}
\end{figure}

\begin{figure}[H]
\centering
  \includegraphics[width=12.5cm]{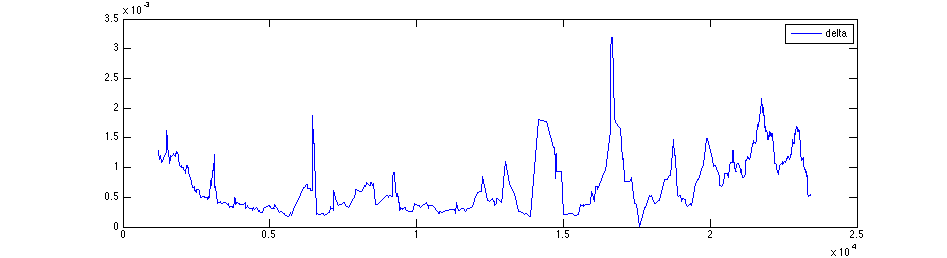}
  \caption{Values of $\hat\delta^n$ when using $10$-minute data $5$-minute before each trading time.}
\end{figure}

\begin{figure}[H]
\centering
  \includegraphics[width=12.5cm]{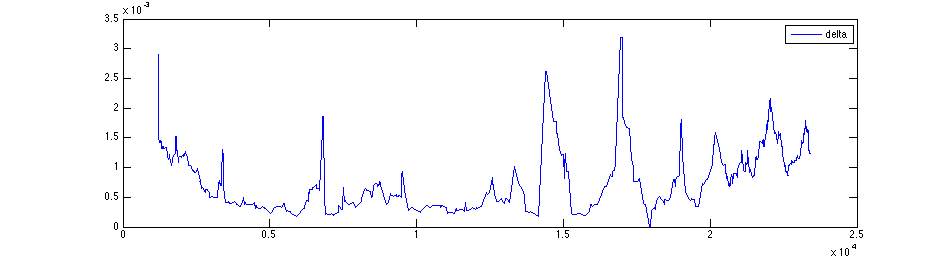}
  \caption{Values of $\hat\delta^n$ when using $10$-minute data $10$-minute before each trading time.}
\end{figure}

\section*{Appendix}

\begin{proof}[Proof of Lemma \ref{ratiodisn}]
We note that $\frac{A(t)}{B(t)} = e^{Y_a(t) - Y_b(t)}$, and $Y_a(t)-Y_b(t)$ is a RBM with mean $\mu_a-\mu_b<0$ and variance $\sigma_a^2+\sigma_b^2$. Thus using Taylor series expansion and Fubini's theorem, we have
\begin{align}\label{mrbm_moments}
E\left[\left(\frac{A(t)}{B(t)}\right)^k\right] & = E\left[e^{k(Y_a(t) - Y_b(t))}\right]  \nonumber\\
& = 1+E\left[\sum_{j=1}^\infty\frac{k^j(Y_a(t) - Y_b(t))^j}{j!}\right] \nonumber \\
& = 1+ \sum_{j=1}^\infty\frac{k^jE[(Y_a(t) - Y_b(t))^j]}{j!},
\end{align}
where from Theorem 1.3 in \cite{AW87},
\begin{align}\label{rbm_moments}
E[(Y_a(t) - Y_b(t))^j] & = E[(Y_a(\infty) - Y_b(\infty))^j]\int_0^\infty g_j(x) F(t;x,0) dx.
\end{align}
Here $Y_a(\infty)-Y_b(\infty)$ is the weak limit of $Y_a(t)-Y_b(t)$ as $t\to\infty$, and $g_k(x)$ is a gamma density with mean $k/2$ and variance $k/4$. We note that from \eqref{rbm_disn}, $Y_a(\infty)-Y_b(\infty)$ follows an exponential distribution with mean $\mu\equiv (\sigma^2_a+\sigma^2_b)/[2(\mu_b-\mu_a)]$, and thus
\begin{align}\label{stationary_moments}
E[(Y_a(\infty)-Y_b(\infty))^j] = j! \mu^j.
\end{align}
Furthermore, the gamma density function is given by
\begin{align}\label{gamma_disn}
g_j(x) = \frac{2^jx^{j-1}}{(j-1)!} e^{-2x}, \ \ x\ge 0.
\end{align}
Putting \eqref{stationary_moments},\eqref{gamma_disn} and \eqref{rbm_moments} into \eqref{mrbm_moments}, using Fubini's theorem again, we have
\begin{align*}
E\left[\left(\frac{A(t)}{B(t)}\right)^k\right] & = 1+ \sum_{j=1}^\infty\int_0^\infty\frac{k^j \mu^j 2^jx^{j-1}}{(j-1)!} e^{-2x} F(t; x, 0) dx \\
& =1 + \int_0^\infty\sum_{j=1}^\infty\frac{k^j \mu^j 2^jx^{j-1}}{(j-1)!} e^{-2x} F(t; x, 0) dx \\
& = 1 + 2k\mu \int_0^\infty e^{2k\mu x - 2x }F(t; x, 0) dx.
\end{align*}
\end{proof}

\begin{proof}[Proof of Proposition \ref{convergence-MRBM}]  For (i), recall that for $t\ge 0$,
\begin{align*}
A(t) & = \exp\left\{X_a(t) + \frac{1}{2}L(t)\right\}, \\
B(t) & = \exp\left\{X_b(t) - \frac{1}{2}L(t)\right\},
\end{align*}
where $L(t) = \sup_{0\le s\le t} (X_a(s)-X_b(s))^-$. So it suffices to show that 
\[
X_a(T_{\delta,n}) - X_b(T_{\delta,n}) = L(T_{\delta,n}).
\]
Now recall that $T_{\delta,0}=0$ and $T_{\delta,n}= \inf\{t\ge 0: X_a(t) -X_b(t)=-2(n-1)\delta\}$, and so $X_a(T_{\delta,n}) - X_b(T_{\delta,n}) = -2(n-1)\delta$, and $X_a(t)-X_b(t)\ge -2(n-1)$ for $t\in[0, T_{\delta,n}]$. Thus 
\[
L(T_{n,\delta}) = 2(n-1)\delta,
\]
and so 
\[
X_a(T_{\delta,n}) - X_b(T_{\delta,n}) = L(T_{\delta,n})=2(n-1)\delta.
\]
To show (ii) and (iii), we first recall that
\begin{align*}
A_\delta(t) & = \exp\left\{X_a(t) + \sum_{n=1}^\infty (n-1)\delta 1_{\{t\in [T_{\delta, (n-1)}, T_{\delta,n})\}}\right\} \\
B_\delta(t) & = \exp\left\{X_b(t) - \sum_{n=1}^\infty (n-1)\delta 1_{\{t\in [T_{\delta, (n-1)}, T_{\delta,n})\}}\right\}.
\end{align*}
Now for $t\in [T_{\delta, n-1}, T_{\delta,n}), n=1, 2, \ldots$, noting that $X_a(T_{\delta,n-1})-X_b(T_{\delta, n-1}) = -2(n-2)\delta$ and that $X_a(t)-X_b(t)$ must be great than $-2(n-1)\delta$, we have that
\begin{align}
L(t) = \sup_{0\le s\le t} (X_a(s)-X_b(s))^- \in [2(n-2)\delta,\ 2(n-1)\delta).
\end{align}
Thus for $t\ge 0$, $A_\delta(t) \ge A(t)$, and
\begin{align}
\sup_{0\le s\le t} \frac{A_\delta(s)}{A(s)} =\sup_{0\le s\le t}\sum_{n=1}^\infty 1_{s\in[T_{\delta, n-1}, T_{\delta, n})}  \exp\left\{(n-1)\delta-\frac{1}{2}L(s)\right\}  \in [e^\delta, 1),
\end{align}
which yields
\begin{align}
\sup_{0\le s\le t}  \frac{A_\delta(s)}{A(s)}\to 1, \ \ \mbox{as $\delta\to0.$}
\end{align}
Similarly, it can be shown that for $t\ge 0,$ $B(t)\ge B_\delta(t)$, and
\begin{align}
\sup_{0\le s\le t}  \frac{B(s)}{B_\delta(s)} \to 1, \ \ \mbox{as $\delta\to0.$}
\end{align}

\end{proof}

\begin{proof}[Proof of Lemma \ref{distn-1}]
Let for $t\ge 0$,
\begin{align*}
X(t) =\begin{pmatrix} X_a(t) \\ X_b(t)\end{pmatrix}= \begin{pmatrix} \alpha \\ \beta\end{pmatrix} + \begin{pmatrix} \mu_a t \\ \mu_b t \end{pmatrix}  + \begin{pmatrix} \sigma_aW_a(t) \\ \sigma_b W_b(t) \end{pmatrix}.
\end{align*}
Then $V_{\delta,1}$ and $U_{\delta,1}$ are the first meeting time and point of $X_a$ and $X_b$.
Let $\theta = \arctan(\sigma_a\sigma_b^{-1})$, and define
\[
M = \begin{pmatrix}
\cos\theta & \sin\theta \\
-\sin\theta & \cos\theta
\end{pmatrix} \begin{pmatrix}
\sigma_a^{-1} & 0 \\
0 & \sigma_b^{-1}
\end{pmatrix},
\]
and for $t\ge 0,$
\begin{align*}
& \check X(t) = M X(t)  \\
&=  \begin{pmatrix} \alpha\sigma_a^{-1}\cos\theta + \beta \sigma_b^{-1}\sin\theta \\ -\alpha\sigma_a^{-1}\sin\theta + \beta \sigma_b^{-1}\cos\theta\end{pmatrix}  + \begin{pmatrix} \mu_a\sigma_a^{-1}\cos\theta + \mu_b \sigma_b^{-1}\sin\theta \\ -\mu_a\sigma_a^{-1}\sin\theta + \mu_b \sigma_b^{-1}\cos\theta\end{pmatrix} t+  \begin{pmatrix}
 W_a(t)\cos\theta +  W_b(t) \sin\theta\\
- W_a(t) \sin\theta+  W_b(t)\cos\theta
\end{pmatrix}.
\end{align*}
Let
\begin{align*}
 \check a & =\alpha\sigma_a^{-1}\cos\theta + \beta \sigma_b^{-1}\sin\theta,\\
\check b & = -\alpha\sigma_a^{-1}\sin\theta + \beta \sigma_b^{-1}\cos\theta,\\
\check \mu_a & =\mu_a\sigma_a^{-1}\cos\theta + \mu_b \sigma_b^{-1}\sin\theta, \\
\check\mu_b & = -\mu_a\sigma_a^{-1}\sin\theta + \mu_b \sigma_b^{-1}\cos\theta, \\
\check B_a(t) & =  W_a(t)\cos\theta +  W_b(t) \sin\theta, \\
\check B_b(t) & = - W_a(t) \sin\theta+  W_b(t)\cos\theta.
\end{align*}
We note that
\begin{align*}
V_{\delta,1} & = \inf\{t\ge 0: X_a(t) = X_b(t)\} \\
& = \inf\left\{t\ge 0: \check X(t) \in \{(x,y): y=0\} \right\} \\
& = \inf\{t\ge 0: \check B_b(t) + \check\mu_b t =-\check b\}.
\end{align*}
Noting that $\check B_b$ is a standard Brownian motion, using Girsonov theorem, and from (5.12) in \cite[Chapter 3.5.C]{ks91}, we have for $t\ge 0,$
\begin{align*}
\Pr(V_{\delta,1}\in dt) & = \frac{|\check b|}{\sqrt{2\pi t^3}} \exp\left\{-\frac{(-\check b - \check \mu_b t)^2}{2t} \right\} dt\nonumber\\
& =  \frac{\alpha-\beta}{\sqrt{2\pi(\sigma_a^2 + \sigma_b^2) t^3} } \exp\left\{-\frac{[\alpha-\beta - (\mu_b -\mu_a)t]^2}{2(\sigma_a^2 + \sigma_b^2) t}\right\} dt.
\end{align*}
Next noting that $\check B_a$ and $V_{\delta,1}$ are independent, we have for $t\ge 0$ and $x\in \mathbb{R}$,
\begin{align*}
 &\Pr(V_{\delta,1}\in dt, U_{\delta,1} \in dx) \\
 & = \Pr(V_{\delta,1}\in dt, X_a(V_{\delta,1}) =X_b(V_{\delta,1})  \in dx) \\
 & = \Pr\left(V_{\delta,1}\in dt, \frac{\check a + \check\mu_a V_{\delta,1} + \check B_a(V_{\delta,1})}{\sigma_a^{-1}\cos\theta + \sigma_b^{-1}\sin\theta} \in dx\right) \\
 & = \Pr\left(\left.\frac{\check a + \check\mu_a t + \check B_a(t)}{\sigma_a^{-1}\cos\theta + \sigma_b^{-1}\sin\theta} \in dx\right|V_{\delta,1}\in dt\right)  \Pr(V_{\delta,1}\in dt) \\
& =  \frac{\sigma_a^{-1}\cos\theta + \sigma_b^{-1}\sin\theta}{\sqrt{2\pi t}} \exp\left\{- \frac{((\sigma_a^{-1}\cos\theta + \sigma_b^{-1}\sin\theta) x- \check a - \check\mu_a t)^2}{2t}\right\} dx \Pr(V_{\delta,1}\in dt) \\
& = \frac{\alpha-\beta}{2\pi t^2 \sigma_a \sigma_b} \exp\left\{-\frac{\left[\frac{\sigma_b}{\sigma_a}(x-\alpha-\mu_at)+\frac{\sigma_a}{\sigma_b}(x-\beta-\mu_bt)\right]^2+[\alpha-\beta - (\mu_b -\mu_a)t]^2}{2(\sigma_a^2 + \sigma_b^2) t}\right\} dx dt,
\end{align*}
where the last equality follows from the identities that
\[
\cos\theta = \frac{\sigma_b}{\sqrt{\sigma^2_a+\sigma^2_b}}, \ \sin\theta = \frac{\sigma_a}{\sqrt{\sigma^2_a+\sigma^2_b}}.
\]
This proves (i). For (ii), we see that for $t\in [0, T_{\delta, n+1}-T_{\delta, n}), n = 1, 2, \ldots$,
\begin{align*}
\tilde{X}_{a,n}(t)\equiv \ln A_\delta(t+T_{\delta,n}) - \ln A_\delta(T_{\delta, n}-) & =  X_a(t+T_{\delta,n})  - X_a(T_{\delta, n}) +  \delta, \\
\tilde{X}_{b,n}(t)\equiv \ln B_\delta(t+T_{\delta,n}) - \ln B_\delta(T_{\delta, n}-) & =  X_b(t+T_{\delta,n})  - X_b(T_{\delta, n}) -  \delta.
\end{align*}
Thus from the strong Markov property of Brownian motions, $\{\tilde{X}_{a,n}(t), t\ge 0\}$ and $\{\tilde{X}_{b,n}(t), t\ge 0\}$ are independent Brownian motions with the initial values $\delta$ and $-\delta$, and the same drifts and variances as $X_a$ and $X_b$. Furthermore, they are independent of $\mathcal{F}_{T_{\delta,n}}$, where
\begin{align}\label{stopping}
\mathcal{F}_t = \sigma\{(X_a(s), X_b(s)), 0\le s\le t\}.
\end{align}
Thus if we let $\tilde{T}_n$ and $\tilde{L}_n$ denote the first meeting time and meeting point of $\tilde{X}_{a,n}$ and $\tilde{X}_{b,n}$, then $\{(\tilde{L}_n, \tilde{T}_n), n=1, 2, \ldots\}$ is an i.i.d. sequence, which is independent of $(U_{\delta,1}, V_{\delta,1})$, and has the same distribution as $(U_{\delta,1}, V_{\delta,1})$ with $\alpha=\delta$ and $\beta=-\delta.$ Finally, noting that $X_a(T_{\delta, n}) - X_b(T_{\delta, n}) = -2(n-1)\delta$, we have that the first meeting time of $\tilde{X}_{a,n}$ and $\tilde{X}_{b,n}$ is given by
\begin{align*}
\tilde{T}_{n} & =  \inf\{t\ge 0: \tilde{X}_{a,n}(t) =\tilde{X}_{b,n}(t)\} \\
 & = \inf\{t\ge 0:  X_a(t+T_{\delta,n})  - X_a(T_{\delta, n}) + \delta  = X_b(t+T_{\delta,n})  - X_b(T_{\delta, n}) - \delta\} \\
 & =  \inf\{t\ge 0:  X_a(t+T_{\delta,n})   - X_b(t+T_{\delta,n}) = - 2n\delta\} \\
 & = T_{\delta, n+1} -T_{\delta,n} \\
 & = V_n,
\end{align*}
and the first meeting point of $\tilde{X}_{a,n}$ and $\tilde{X}_{b,n}$ is given by
\begin{align*}
\tilde{L}_n & = \tilde{X}_{a,n}(\tilde{T}_n) = \ln A_\delta(T_{\delta,n+1}-) - \ln A_\delta(T_{\delta,n}-) = \ln P_{n+1} - \ln P_n = U_n.
\end{align*}
To summarize, we have shown that $\{(U_n, V_n), n= 2, 3, \ldots\}$ is an i.i.d. sequence, which is independent of $(V_{\delta,1}, U_{\delta,1})$, and has the same distribution as $(V_{\delta,1}, U_{\delta,1})$ with $\alpha=\delta$ and $\beta=-\delta.$
\end{proof}

\begin{proof}[Proof of Lemma \ref{thm2_moments}] Assume $A(0)=e^\delta$ and $B(0)=e^{-\delta}$. Then $(U_{\delta,1}, V_{\delta,1})$ has the same distribution as $(U_\delta,V_\delta).$
Let
\begin{eqnarray*}
Y_a(t) &=& \exp \left\{ {\theta _1}X_a(t) - ({\theta _1}{\mu _a} + \frac{1}{2}\theta _1^2\sigma _a^2)t \right\},\\
Y_b(t) &=& \exp \left\{ {\theta _2}X_b(t) - ({\theta _2}{\mu _b} + \frac{1}{2}\theta
_2^2\sigma _b^2)t \right\},
\end{eqnarray*}
where $\theta_1$ and $\theta_2$ are arbitrary real numbers. Then $\{Y_a(t), t\geq 0\}$ and $\{Y_b(t), t\geq 0\}$ are independent, and $\{(Y_a(t),Y_b(t)), t\geq 0\}$ is a $\{\mathcal{F}_t\}_{t\ge 0}$ martingale (see the beginning of Section 5 of Chapter 7 in \cite{kt75}), where $\mathcal{F}_t$ is defined in \eqref{stopping}. We also note that $V_{\delta,1}$ is an $\{\mathcal{F}_t\}_{t\ge 0}$ stopping time with finite mean and variance ($V_{\delta,1}$ follows IG distribution from Lemma \ref{distn-1}). Hence optional stopping theorem yields
\begin{equation*}
(E(Y_a(V_{\delta,1})), E(Y_b(V_{\delta,1}))) = (E(Y_a(0)), E(Y_b(0))),
\end{equation*}
and so $E[Y_a(V_{\delta,1})Y_b(V_{\delta,1})]  = E[Y_a(0)Y_b(0)]$. More precisely, we have
\begin{equation*}
E\left\{ {\exp \{ ({\theta _1} + {\theta _2})U_{\delta,1} - ({\theta _1}{\mu _a} + \frac{1}{2}\theta _1^2\sigma
_a^2 + {\theta _2}{\mu _b} + \frac{1}{2}\theta _2^2\sigma _b^2)V_{\delta,1}\} } \right\} = \exp \{ [{\theta _1} -
{\theta _2}]\delta\}.
\end{equation*}
Let
\begin{eqnarray*}
&&{\theta _1} + {\theta _2} =  s, \\
&&{\theta _1}{\mu _a} + \frac{1}{2}\theta _1^2\sigma _a^2 + {\theta _2}{\mu _b} + \frac{1}{2}\theta
_2^2\sigma _b^2 = - t.
\end{eqnarray*}
Solving $\theta_1$ and $\theta_2$ in terms of $s$ and $t$, we obtain
\begin{align*}
{\theta _1(s,t)} &= \frac{{({\mu _b} - {\mu _a} + s\sigma _b^2) \pm \sqrt {{{({\mu _b} - {\mu _a} + s\sigma
_b^2)}^2} - (\sigma _a^2 + \sigma _b^2)({s^2}\sigma _b^2 + 2t + 2s{\mu _b})} }}{{\sigma _a^2 + \sigma _b^2}}, \\
{\theta _2(s,t)} &=   s - \theta_1(s,t).
\end{align*}
Letting $s=0$, and noting that $V_{\delta,1}$ follows inverse Gaussian distribution (see \eqref{meeting-time}), the moment generating function of $V_{\delta,1}$ is
\[
E(\exp(t V_1)) = \frac{{({\mu _b} - {\mu _a}) - \sqrt {{{({\mu _b} - {\mu _a})}^2} - 2t (\sigma _a^2 + \sigma _b^2)} }}{{\sigma _a^2 + \sigma _b^2}}.
\]
Thus the solutions of $\theta_1(s,t)$ should be as in (\ref{theta1}), and the moment generating function $\phi(s,t)$ of $(U_\delta,V_\delta)$ is given by \eqref{mgf} with $\theta(s,t)$ instead of $\theta_1(s,t)$ as in \eqref{theta1}. To compute the moments, we first need some simple results about $\theta_1(s,t)$ as follows.
\begin{align*}
&{\theta }(0,0)=0,\\
&{\left. {\frac{{\partial {\theta}(s,t)}}{{\partial t}}} \right|_{s = t = 0}}= \frac{1}{{{\mu _b} -
{\mu _a}}},\;\;{\left. {\frac{{\partial {\theta}(s,t)}}{{\partial s}}} \right|_{s = t = 0}}= \frac{{{\mu
_b}}}{{{\mu _b} - {\mu _a}}}, \\
 &{\left. {\frac{{{\partial ^2}{\theta}(s,t)}}{{\partial {t^2}}}} \right|_{s = t =
0}}=\frac{{\sigma _a^2 + \sigma _b^2}}{{{{({\mu _b} - {\mu _a})}^3}}},\;\;{\left. {\frac{{{\partial ^2}{\theta}(s,t)}}{{\partial {s^2}}}}
\right|_{s = t = 0}} = \frac{{\mu _b^2\sigma _a^2 + \mu _a^2\sigma _b^2}}{{{{({\mu _b} - {\mu_a})}^3}}},\\
&{\left. {\frac{{{\partial ^2}{\theta}(s,t)}}{{\partial s\partial t}}} \right|_{s = t = 0}} =
\frac{{{\mu _b}\sigma _a^2 + {\mu _a}\sigma _b^2}}{{{{({\mu _b} - {\mu _a})}^3}}}, \ \  \left.\frac{\partial^3 \theta(s, t)}{\partial t^2\partial s}\right|_{s=t=0} = \frac{3(\sigma^2_a\mu_b+\sigma^2_b\mu_a)(\sigma^2_a+\sigma^2_b)}{(\mu_b-\mu_a)^5}.
\end{align*}
Therefore,
\begin{align*}
E(V_\delta) &=  \left.\frac{\partial \phi(s,t)}{\partial t} \right|_{s=t=0}=\frac{\partial }{{\partial t}}{\left. {\exp \{ [2{\theta}(s,t)  - s] \delta \} } \right|_{s = 0,t = 0}}\\
&= \frac{2\delta }{{{\mu _b} - {\mu _a}}}.
\end{align*}
Similarly, we obtain
\begin{align*}
E({U_\delta} ) & = \frac{{\delta ({\mu _b} + {\mu _a})}}{{{\mu _b} - {\mu _a}}}\\
E(V^2_\delta) & = \frac{{{4\delta ^2}}}{{{{({\mu _b} - {\mu _a})}^2}}} + \frac{{2(\sigma _a^2 + \sigma _b^2)\delta }}{{{{({\mu _b} - {\mu _a})}^3}}}\\
E({U_\delta}^2) & = \frac{{{\delta ^2}(\mu _a + \mu _b)^2}}{{{{({\mu _b} - {\mu _a})}^2}}} + \frac{{2(\mu _b^2\sigma _a^2 + \mu _a^2\sigma _b^2)\delta }}{{{{({\mu _b} - {\mu _a})}^3}}}\\
E\left( {{U_\delta}V_\delta} \right) & =\frac{{{2\delta ^2}({\mu _b} + {\mu _a})}}{{{{({\mu _b} - {\mu _a})}^2}}} + \frac{{2({\mu _b}\sigma _a^2
+ {\mu _a}\sigma _b^2)\delta }}{{{{({\mu _b} - {\mu _a})}^3}}}.
\end{align*}
Finally, for $k,l\in\mathbb{N}\cup\{0\}$ and $k+l\ge 1$, \eqref{moments} follows by noting that
\begin{align*}
E(U^k_\delta V^l_\delta) = \left.\frac{\partial^{k+l}\phi(s,t)}{\partial s^k\partial t^l} \right|_{s=t=0} = \delta \left(\phi(s,t) \left.\frac{\partial^{k+l}}{\partial s^k\partial t^l} (2\theta(s,t)-s)\right)\right|_{s=t=0} + o(\delta),
\end{align*}
where $o(\delta)\to 0$ as $\delta\to0.$
\end{proof}

\begin{proof}[Proof of Theorem \ref{th:asy}]
From Brown and Solomon~\cite{bs75}, we have the following results for a renewal reward process generated by $\{(U_{\delta,n},V_{\delta,n}), n \ge 1\}$:
\[ E(Z_\delta(t)) = m t + \mathcal{O}(1),\]
where
\[ m = \frac{E(U_{\delta})}{E(V_{\delta})}.
\]
Using the results of Lemma \ref{thm2_moments} in the above equation, we get Equation \eqref{eq:m}. The same paper also states that
\[ \mbox{Var}(Z_\delta(t)) = s t  + \mathcal{O}(1),\]
where
\[ s = \frac{E(V_{\delta}^2)E(U_{\delta})^2}{E(V_{\delta})^3} - \frac{2E(U_{\delta}V_{\delta})E(U_{\delta})}{E(V_{\delta})^2} + \frac{E(U_{\delta}^2)}{E(V_{\delta})}, \]
Substituting the moments of $(U_{\delta},V_{\delta})$ from Lemma \ref{thm2_moments} into the above equation and simplifying, we get Equation \eqref{eq:s}.
\end{proof}

\begin{proof}[Proof of Theorem \ref{th:conv}]
Consider an arbitrary nonnegative sequence $\{\delta_m\}_{m\ge 1}$ such that $\delta_m \to 0$ as $m\to\infty.$ Define for $m, n\ge 1$,
\begin{align*}
\tilde U_{m,n} & = \sqrt{\delta_m}(U_{\delta_m, n} - E(U_{\delta_m,n})), \\
\tilde V_{m,n} & = \sqrt{\delta_m}(V_{\delta_m, n} - E(V_{\delta_m,n})).
\end{align*}
We note that for each $m$, $\{(\tilde U_{m,n}, \tilde V_{m,n}), {n\ge 2}\}$ is an i.i.d. sequence. Furthermore,
\[
\sum_{n=1}^{\lfloor \frac{t}{\delta^2_m}\rfloor} \Var(\tilde U_{m,n}) \to \frac{{2(\mu _b^2\sigma _a^2 + \mu _a^2\sigma _b^2)t }}{{{{({\mu _b} - {\mu _a})}^3}}}, \ \mbox{and} \ \sum_{n=1}^{\lfloor \frac{t}{\delta^2_m}\rfloor} \Var(\tilde V_{m,n}) \to \frac{{2(\sigma _a^2 + \sigma _b^2)t }}{{{{({\mu _b} - {\mu _a})}^3}}}, \ \mbox{as $m\to\infty.$}
\]
We claim that $\{(\tilde U_{m,n}, \tilde V_{m,n}), {m\ge 1, 1\le n \le \lfloor \frac{t}{\delta^2_m}\rfloor}\}$ satisfies Lindeberg condition, i.e., for any $\epsilon >0$,
\begin{align}\label{claim}
\sum_{n=1}^{\lfloor \frac{t}{\delta^2_m}\rfloor} E\left( \tilde U_{m,n}^2 1_{\{|\tilde U_{m,n}|\ge \epsilon\}}\right) \to 0, \ \mbox{and} \ \sum_{n=1}^{\lfloor \frac{t}{\delta^2_m}\rfloor} E\left( \tilde V_{m,n}^2 1_{\{|\tilde V_{m,n}|\ge \epsilon\}}\right) \to 0, \ \mbox{as $m\to\infty.$}
\end{align}
We will prove \eqref{claim} at the end of this proof. Thus from \cite[Theorem 18.2]{b99}, letting
\[
u_m(t) = \sum_{n=1}^{\lfloor \frac{t}{\delta^2_m}\rfloor} \tilde U_{m,n},
\ \mbox{and} \ v_m(t) = \sum_{n=1}^{\lfloor \frac{t}{\delta^2_m}\rfloor} \tilde V_{m,n},
\]
then
\[
(u_m, v_m) \Go W, \ \mbox{as $m\to\infty.$}
\]
where $W$ is a two dimensional Brownian motion with drift $0$ and covariance matrix
\[\frac{2}{({\mu _b} - {\mu _a})^3}\begin{pmatrix}
\mu _b^2\sigma _a^2 + \mu _a^2\sigma _b^2 & {\mu _b}\sigma _a^2 + {\mu _a}\sigma _b^2 \\
{\mu _b}\sigma _a^2 + {\mu _a}\sigma _b^2 & \sigma _a^2 + \sigma _b^2
\end{pmatrix}.
 \]
Next from \cite[Theorem 1]{iw71} and \cite[Corollary 3.33]{js03}, if
\begin{align*}
\tilde N_m(t) & = (E(V_{\delta_m,1}))^{3/2} \left(N_{\delta_m}\left({t}/{\delta_m}\right) - \frac{t}{\delta_m E(V_{\delta_m,1})}\right) \\
& = \left(\frac{2\delta_m}{\mu_b -\mu_a}\right)^{3/2} \left(N_{\delta_m}\left({t}/{\delta_m}\right) - \frac{(\mu_b-\mu_a)t}{2\delta_m^2}\right),
\end{align*}
then $(u_m, v_m, \tilde N_m) \Go (W_1, W_2, -W_2)$ as $m\to\infty$, where $W_1$ and $W_2$ are the first and second components of the Brownian motion $W$. Finally, we note that
\[
\clz_{\delta_m}(t) = \frac{1}{\sqrt{s}}\left[u_m\left(\delta_m^2N_{\delta_m}(t/\delta_m)\right) + \frac{\mu_b+\mu_a}{\mu_b-\mu_a} \left(\frac{\mu_b -\mu_a}{2}\right)^{3/2} \tilde N_m(t)\right].
\]
Furthermore, observing that
\[
\delta_m^2N_{\delta_m}(t/\delta_m) = \delta_m^2\left[\frac{\tilde N_m(t)}{(E(V_{\delta_m,1}))^{3/2}} + \frac{(\mu_b-\mu_a)t}{2\delta_m^2}\right] \to \frac{(\mu_b-\mu_a)t}{2}, \ \mbox{as $m\to\infty,$}
\]
we have that
\[
\clz_{\delta_m}(\cdot) \Go \frac{W_1(\frac{\mu_b-\mu_a}{2}\ \cdot) + \frac{\mu_b+\mu_a}{\mu_b-\mu_a} \left(\frac{\mu_b -\mu_a}{2}\right)^{3/2}
W_2(\cdot)}{\sqrt{s}},
\]
and it is easy to check that the weak limit on the right hand side is a standard Brownian motion. Consequently, $\clz_\delta$ converges weakly to a standard Brownian motion as $\delta\to 0$. At last, we give the proof of the claim given in \eqref{claim}. The proofs for $\tilde V_{m,n}$ and $\tilde U_{m,n}$ are similar, and we only consider $\tilde V_{m,n}$. We first note that from  Lemma \ref{thm2_moments},
\begin{align*}
E(V_{\delta_m,1}|A(0), B(0)) & = \frac{\ln A(0) - \ln B(0)}{\mu_b -\mu_a}, \\
\Var(V_{\delta_m,1}| A(0), B(0)) & = \frac{(\ln A(0) - \ln B(0))(\sigma^2_a + \sigma^2_b)}{(\mu_b - \mu_a)^3},
\end{align*}
and using conditional expectations, we have that for some $b_0\in (0,\infty)$,
\begin{align*}
\Var(V_{\delta_m,1})& = E(\Var(V_{\delta_m,1}|A(0), B(0))) + \Var(E(V_{\delta_m,1}|A(0), B(0))) \\
&  \le b_0 \left( E[\ln(A(0)/B(0))]+E[\ln^2(A(0)/B(0))] \right) < \infty.
\end{align*}
Next using Markov inequality, Holder's inequality and \eqref{moments}, we have for some $c_0\in (0,\infty),$
\begin{align*}
& \sum_{n=1}^{\lfloor \frac{t}{\delta^2_m}\rfloor} E\left( \tilde V_{m,n}^2 1_{\{|\tilde V_{m,n}|\ge \epsilon\}}\right)\\
& \le E(\tilde V_{m,1}^2) + \lceil \frac{t}{\delta^2_m}\rceil \sqrt{E(\tilde V_{m,2}^4) P(|\tilde V_{m,2}|\ge \epsilon)} \\
& \le \delta_m \Var(V_{\delta_m,1}) + \lceil \frac{t}{\delta^2_m}\rceil \sqrt{E(\tilde V_{m,2}^4) \epsilon^{-2} E(\tilde V_{m,2}^2)} \\
& \le \delta_m \Var(V_{\delta_m,1})+  \epsilon^{-1} \lceil \frac{t}{\delta^2_m}\rceil \delta_m^{3/2}\sqrt{E[(V_{\delta_m,2}-E(V_{\delta_m,2}))^4] \Var(V_{\delta_m,2})} \\
& \le \delta_m \Var(V_{\delta_m,1}) +\epsilon^{-1} \lceil \frac{t}{\delta^2_m}\rceil \delta_m^{3/2} \sqrt{c_0 \delta_m^2} \\
& \to 0, \ \mbox{as $m\to\infty$.}
\end{align*}
\end{proof}

\begin{proof}[Proof of Lemma \ref{estimation1}] For convenience, we omit the superscript $n$ for the estimators of $\mu_a, \mu_b, \sigma_a, \sigma_b$ and $\delta$.
Using the moments in Lemma \ref{thm2_moments}, we consider the following equations.
\begin{align}
x_1 &= \frac{2\hat\delta }{{{\hat\mu _b} - {\hat\mu _a}}}\label{moment-eqn-1}\\
x_2 &= \frac{{\hat\delta ({\hat\mu _b} + {\hat\mu _a})}}{{{\hat\mu _b} - {\hat\mu _a}}}\\
x_3 &= \frac{{{4\hat\delta ^2}}}{{{{({\hat\mu _b} - {\hat\mu _a})}^2}}} + \frac{{2(\hat\sigma _a^2 + \hat\sigma _b^2)\hat\delta }}{{{{({\hat\mu _b} - {\hat\mu _a})}^3}}}\\
x_4 &= \frac{{{\hat\delta ^2}(\hat\mu _a + \hat\mu _b)^2}}{{{{({\hat\mu _b} - {\hat\mu _a})}^2}}} + \frac{{2(\hat\mu _b^2\hat\sigma _a^2 + \hat\mu _a^2\hat\sigma _b^2)\hat\delta }}{{{{({\hat\mu _b} - {\hat\mu _a})}^3}}}\\
x_5 &= \frac{{{2\hat\delta ^2}({\hat\mu _b} + {\hat\mu _a})}}{{2{{({\hat\mu _b} - {\hat\mu _a})}^2}}}
+\frac{{2({\hat\mu _b}\hat\sigma _a^2 + {\hat\mu _a}\hat\sigma _b^2)\hat\delta }}{{{{({\hat\mu _b} -
{\hat\mu _a})}^3}}}. \label{moment-eqn-2}
\end{align}
Next we solve the above equations for $\hat\mu_a$, $\hat\mu_b$, $\hat\sigma_a$, $\hat\sigma_b$, $\hat\delta$ in terms of $x_k, k =1,2,\ldots, 5$. Let
\begin{eqnarray}
{y_1} &=& \frac{2x_2}{x_1} = {\hat\mu _b} + {\hat\mu _a} \nonumber\label{unique1}\\
{y_2} &=& \frac{x_3-x_1^2}{x_1} = \frac{{\hat\sigma _a^2 + \hat\sigma _b^2}}{{{{({\hat\mu _b} - {\hat\mu _a})}^2}}}\nonumber\\
{y_3} &=& \frac{x_4-x_2^2}{x_1} = \frac{{\hat\mu _b^2\hat\sigma _a^2 + \hat\mu _a^2\hat\sigma _b^2}}{{{{({\hat\mu _b} - {\hat\mu _a})}^2}}}\nonumber\\
{y_4} &=& \frac{x_5-x_1 x_2}{x_1} = \frac{{{\hat\mu _b}\hat\sigma _a^2 + {\hat\mu _a}\hat\sigma
_b^2}}{{{{({\hat\mu _b} - {\hat\mu _a})}^2}}}.\nonumber
\end{eqnarray}
We then note that
\begin{equation*}\label{unique2}
{y_1^2 - 4\frac{{{y_1}{y_4} - {y_3}}}{{{y_2}}}} = (\hat\mu_b - \hat\mu_a)^2.
\end{equation*}
Letting $\hat\mu_b > \hat\mu_a$, we obtain
\begin{eqnarray*}
{\hat\mu _a} &=& \frac{{{y_1} - {\sqrt{ {y_1^2 - 4\frac{{{y_1}{y_4} - {y_3}}}{{{y_2}}}}}}}}{2}\\
{\hat\mu _b} &=& \frac{{{y_1} + {\sqrt{{y_1^2 - 4\frac{{{y_1}{y_4} - {y_3}}}{{{y_2}}}}}}}}{2}
\end{eqnarray*}
and
\begin{eqnarray*}
{\hat\sigma _a} &=& \sqrt {({y_4} - {\hat\mu_a}{y_2})({\hat\mu_b} - {\hat\mu_a})}, \\
{\hat\sigma _b} &=& \sqrt {({\hat\mu_b}{y_2} - {y_4})({\hat\mu_b} - {\hat\mu_a})},\\
\hat\delta  &=& ({\hat\mu_b} - {\hat\mu_a}){x_1}.
\end{eqnarray*}
To see the above estimators are well-defined, we only need to show \eqref{well-define}.
We first note that
\begin{equation*}
y_1^2 - \frac{4({{y_1}{y_4} - {y_3}})}{{{y_2}}} = \frac{4}{{x_1^2({x_3} - x_1^2)}}\left[ {x_2^2({x_3} - x_1^2)
- 2{x_1}{x_2}({x_5} - {x_1}{x_2}) + x_1^2({x_4} - x_2^2)} \right].
\end{equation*}
It is clear that
\begin{eqnarray*}
x_1^2 &=& \left(\sum\limits_{i = 1}^n \frac{{v_i}}{n}\right)^2 >0,\\
{x_3} - x_1^2 &=& \frac{n \sum\limits_{i = 1}^n {{v_i^2}} - \sum\limits_{i = 1}^n {{v_i}}}{n^2} > 0.
\end{eqnarray*}
We next note that
\begin{align*}
&{x_2^2({x_3} - x_1^2) - 2{x_1}{x_2}({x_5} - {x_1}{x_2}) + x_1^2({x_4} - x_2^2)}\\
&\geq 2 x_1 x_2 \sqrt{({x_3} - x_1^2)({x_4} - x_2^2)} - 2{x_1}{x_2}({x_5} - {x_1}{x_2}) \\
&= 2 x_1 x_2 (\sqrt{({x_3} - x_1^2)({x_4} - x_2^2)} - ({x_5} - {x_1}{x_2})) \\
&= 2\frac{{\sum {{v_i}} }}{n}\frac{{\sum {{u_i}} }}{n}\left( {\sqrt {\left( {\frac{{\sum {v_i^2} }}{n} -
{{\left( {\frac{{\sum {{v_i}} }}{n}} \right)}^2}} \right)\left( {\frac{{\sum {u_i^2} }}{n} - {{\left(
{\frac{{\sum {{u_i}} }}{n}} \right)}^2}} \right)}  - \left( {\frac{{\sum {{v_i}{\Delta p_i}} }}{n} -
\frac{{\sum
{{v_i}} }}{n}\frac{{\sum {{u_i}} }}{n}} \right)} \right)\\
&= 2\frac{{\sum {{v_i}} }}{n}\frac{{\sum {{u_i}} }}{n}\left( {\sqrt {\frac{{\sum {{{\left( {{v_i} - \sum
{{v_i}} /n} \right)}^2}} }}{n}\frac{{\sum {{{\left( {{u_i} - \sum {{u_i}} /n} \right)}^2}} }}{n}}  - \left(
{\frac{{\sum {{v_i}{u_i}} }}{n} - \frac{{\sum {{v_i}} }}{n}\frac{{\sum {{u_i}} }}{n}} \right)} \right)\\
&\geq 2\frac{{\sum {{v_i}} }}{n}\frac{{\sum {{u_i}} }}{n}\left( {\frac{{\sum {\left( {{v_i} - \sum {{v_i}}
/n} \right)\left( {{u_i} - \sum {{u_i}} /n} \right)} }}{n} - \left( {\frac{{\sum {{v_i}{u_i}} }}{n} -
\frac{{\sum {{v_i}} }}{n}\frac{{\sum {{u_i}} }}{n}} \right)} \right)\\
&= 0.
\end{align*}
This shows the first inequality in \eqref{well-define}. To show the last two inequalities in \eqref{well-define}, we observe that
\begin{eqnarray*}
{y_4} - {\hat\mu_a}{y_2} &=& \frac{{{y_2}{{\sqrt{y_1^2 - \frac{4({{y_1}{y_4} - {y_3}})}{{{y_2}}}}}}}}{2} + \left( {{y_4} - \frac{{{y_1}{y_2}}}{2}} \right),\\
{\hat\mu_b}{y_2} - {y_4} &=& \frac{{{y_2}{{\sqrt{y_1^2 - \frac{4({{y_1}{y_4} - {y_3}})}{{{y_2}}}}}}}}{2} - \left( {{y_4} - \frac{{{y_1}{y_2}}}{2}} \right).
\end{eqnarray*}
Hence it suffices to show
\begin{equation*}
\frac{{y_2^2\left( {y_1^2 - \frac{4({{y_1}{y_4} - {y_3}})}{{{y_2}}}} \right)}}{4} \ge {\left( {{y_4} -
\frac{{{y_1}{y_2}}}{2}} \right)^2}.
\end{equation*}
After simplifying above inequality, it suffices to show that $y_2y_3 \ge y_4^2$. Note that
\begin{equation*}
y_2y_3 \ge y_4^2
\end{equation*}
is equivalent to
\[
(x_3-x_1^2)(x_4-x_2^2) \ge (x_5-x_1x_2)^2,
\]
and the latter one is proved above. This completes the proof of \eqref{well-define}. Next from the construction of the estimators, we see that they are the unique solutions of \eqref{moment-eqn-1} -- \eqref{moment-eqn-2}. Using the strong law of large numbers and the continuous mapping theorem, we have \eqref{consistent}. Finally, the central limit theorem for $\hat\Theta$ follows immediatly from Delta method (see \cite{cb01}) and the central limit theorem for $(x_1, x_2, \ldots, x_5)$, i.e., $$\sqrt{n}[(x_1, x_2, x_3, x_4, x_5) - E(x_1, x_2, x_3, x_4, x_5)]\Go \mathcal{N}_5(0,\Sigma),$$ where $\Sigma$ is the covariance matrix of $(V_\delta, U_\delta, V^2_\delta, U^2_\delta, U_\delta V_\delta).$
\end{proof}

\bibliographystyle{amsplain}
\bibliography{reference.bib}

\end{document}